\documentclass[11pt,a4paper]{article}
\usepackage[margin=2.6cm]{geometry}
\usepackage[utf8]{inputenc}
\usepackage[T1]{fontenc}
\usepackage{lmodern}
\usepackage{microtype}
\usepackage{amsmath,amssymb,amsthm,mathtools}
\usepackage{enumitem}
\usepackage{array}
\usepackage{tabularx}
\usepackage{hyperref}
\usepackage{url}
\hypersetup{colorlinks=true,linkcolor=blue,citecolor=blue,urlcolor=blue}

\newtheorem{theorem}{Theorem}[section]
\newtheorem{lemma}[theorem]{Lemma}
\newtheorem{proposition}[theorem]{Proposition}
\newtheorem{corollary}[theorem]{Corollary}
\newtheorem{definition}[theorem]{Definition}
\newtheorem{remark}[theorem]{Remark}
\newtheorem{example}[theorem]{Example}

\theoremstyle{remark}
\newtheorem*{remark*}{Remark}

\newcommand{\CS}{\mathrm{CS}}
\newcommand{\Cmcf}[2]{\mathcal{C}^{\mathrm{mcf}}_{#1,#2}}
\newcommand{\D}{\mathcal{D}}
\newcommand{\out}{\mathrm{out}}
\newcommand{\tuple}[1]{\vec{#1}}
\newcommand{\blank}{\square}
\newcommand{\eps}{\varepsilon}
\newcommand{\obs}{\mathrm{obs}}

\providecommand{\keywords}[1]{\par\smallskip\noindent\textbf{Keywords. }#1\par}

\title{The Value of Finite Observation in Positive-Data Learning of Multiple Context-Free Languages}
\author{Takayuki Kuriyama\\\texttt{growup.kuriyama@gmail.com}}
\date{}

\usepackage{tikz}

\newcommand{\cutpos}[1]{%
  \raisebox{-0.7ex}{%
    \tikz[baseline=(char.base)]{%
      \node[draw,circle,inner sep=0.4pt,font=\scriptsize]
        (char) {\ensuremath{#1}};%
    }%
  }%
}

\begin{document}
\setlength{\abovedisplayskip}{6pt plus 2pt minus 3pt}
\setlength{\belowdisplayskip}{6pt plus 2pt minus 3pt}
\setlength{\abovedisplayshortskip}{3pt plus 1pt minus 2pt}
\setlength{\belowdisplayshortskip}{3pt plus 1pt minus 2pt}
\maketitle

\begin{abstract}
Positive data can show that two tuple occurrences share a successful sentence
context without certifying that they are safely interchangeable.  We study
finite compositional observations, represented by finite-monoid
homomorphisms, as semantic side information for learning bounded-fan-out
multiple context-free languages from positive data.

For every fixed fan-out bound \(f\) and supplied finite observation \(h\), we
define observation-guarded tuple substitutability and give a canonical
set-driven learner that exactly reconstructs every language in the full
semantic slice
\[
  \{L\in f\text{-}\mathrm{MCFL}:L\text{ is }(f,h)\text{-tuple-substitutable}\}.
\]
Under a fixed branching cap, hypothesis construction is polynomial.  More
generally, polynomial exposure of a presentation implies polynomial
characteristic data; binary single-spine presentations provide an explicit
sufficient condition.

We then vary how much observation information is available.  A fixed bound on
the size of an unknown observation can be compiled into a universal finite
refinement, restoring identifiability slicewise, whereas the unbounded
latent-observation union is not identifiable.  Moreover, universal finite
refinements have an unavoidable exponential dependence on the observation
bound, and a separate deletion obstruction shows that no family of set-driven
identifiers can have characteristic data polynomial jointly in that bound and
presentation size.  Intrinsic observation size alone therefore does not
determine latent-slice data complexity: the quantitative behavior also depends
on whether the witnessing semantic slice is supplied or must be resolved
inside a larger ambient class.
\end{abstract}


\section{Introduction}
\label{sec:introduction}

Positive-data identification in the limit originates with Gold and was
developed further in the formal-language setting by Angluin
\cite{Gold1967,Angluin1980}.  A central difficulty is that positive evidence can
show that two expressions occur in a common successful environment without
showing that they are interchangeable in every successful environment
\cite{ClarkEyraud2007}.  Broad
language families are therefore not identifiable from positive data alone,
while distributional and structural restrictions can make finite positive
witnesses sufficient
\cite{ClarkEyraud2007,Yoshinaka2008,Yoshinaka2009,Yoshinaka2011,Kanazawa1998}.
For multiple context-free languages, Yoshinaka's multidimensional
substitutability framework is the closest predecessor of the tuple-indexed
construction used below \cite{Yoshinaka2009,Yoshinaka2011}.

This paper asks a complementary question: \emph{what finite semantic side
information is sufficient to license safe positive-data
merges\footnote{Here a merge means treating two substrings or tuples as having
the same role and identifying them with a single state or nonterminal.}, and
what is the cost of supplying or withholding that information?}  We represent the
side information by a finite compositional observation
\(h:\Sigma^*\to M\)
and require substitutability only between tuples having equal componentwise
observation values and sharing an accepting named sentence context.  The
observation therefore acts as a finite guard on distributional merging.  The
semantic promise is stated at the language level, while the learner operates
only on concrete positive witnesses.

The distinction is substantive in the MCFG setting.  Tuple components may be
permuted and separated by terminal material, so reconstruction must recover
occurrence-sensitive child placement rather than only a one-dimensional
constituent boundary.  More importantly for the present paper, the observation
is not merely a typing convenience: the results below compare three
information interfaces, according to whether a target observation is supplied,
merely known to be bounded, or left completely latent.

\paragraph{Main results: an observation-information hierarchy.}
The principal results can be read in three steps.
At a high level, the results form the following hierarchy.
\begin{table}[t]
\centering
\small
\caption{Learning behavior under three observation interfaces.}
\label{tab:observation-hierarchy}
\begin{tabularx}{\textwidth}{>{\raggedright\arraybackslash}p{0.20\textwidth}>{\raggedright\arraybackslash}p{0.34\textwidth}>{\raggedright\arraybackslash}X}
\hline
Observation information & Positive-data identification & Quantitative behavior \\
\hline
supplied \(h\) & yes, on the full slice \(\Cmcf{f}{h}\) & polynomial construction at fixed rank; polynomial data for polynomially exposable presentations, in particular single-spine \\
\(|M|\le k\) known & yes, slicewise via \(H_{\Sigma,k}\) & polynomial slicewise under single-spine, but not uniformly polynomial in \(k\) \\
no bound on \(h\) & no & non-identifiable \\
\hline
\end{tabularx}
\end{table}
\begin{enumerate}[label=(\Roman*),leftmargin=*]
\item \textbf{Supplied observation.}
For fixed \(f,h\), the canonical set-driven learner has a finite positive
characteristic sample for every target \(L\in\Cmcf{f}{h}\), and every finite
positive sample for \(L\) containing that characteristic sample yields a
hypothesis with language \(L\)
(Theorem~\ref{thm:exact-reconstruction} and
Corollary~\ref{cor:semantic-identification}).
No common target rule-rank bound is required for qualitative identification.
With a fixed branching cap, hypothesis construction is polynomial
(Theorem~\ref{thm:poly}); polynomial exposure implies polynomial-size
characteristic samples
(Theorem~\ref{thm:polynomial-exposure-implies-data}), and single-spine
presentations give the concrete bound \(O_f(m^{3f}|G|^3)\)
(Theorem~\ref{thm:single-spine-polynomial-exposure}).

\item \textbf{Bounded latent observation.}
If only an observation-size bound \(k\) is known, finitely many candidate
observations can be compiled into one universal refinement \(H_{\Sigma,k}\),
restoring identification on each fixed slice
(Theorem~\ref{thm:bounded-union-identifiable}) and slicewise polynomial time
and data under single-spine
(Theorem~\ref{thm:bounded-union-single-spine-poly}).

\item \textbf{Unbounded or uniformly varying observation.}
Without any bound, the latent-observation union is not identifiable
(Theorem~\ref{thm:no-advice-nonidentifiability}).  Quantitatively, explicit
universal observations already have an exponential-in-\(k\) obstruction
(Corollary~\ref{cor:universal-observation-size-lower-bound} and
Proposition~\ref{prop:universal-learner-exponential-k-data}); separately, large
bounded latent slices contain deletion families that rule out any characteristic-
data bound polynomial jointly in the observation bound and presentation size
(Theorem~\ref{thm:no-uniform-polynomial-k-data}).  Thus intrinsic observation
size alone does not control latent-slice data complexity.
\end{enumerate}

\paragraph{Relation to earlier distributional learning.}
Yoshinaka's multidimensional substitutability is the closest predecessor of
the tuple construction \cite{Yoshinaka2009,Yoshinaka2011}.  The semantic
conditions differ: \(L_3\) belongs to a fixed-observation slice while failing
Yoshinaka's condition already at dimension two
(Proposition~\ref{prop:l3-not-yoshinaka-2d}).  Section~\ref{sec:comparison}
separates the information interfaces and the role of rule-rank parameters in
detail.  The qualitative theorem here concerns the full semantic fixed-fan-out
slice; binary, branching, and single-spine assumptions enter only the
quantitative presentation results.

\section{Preliminaries and the Target Class}
\label{sec:prelim}

\subsection{Positive-data learning}

\begin{definition}[Text and identification]
To handle the empty language, fix a pause symbol
\(\#\notin\Sigma\).%
\footnote{Allowing the pause symbol makes
\(\#,\#,\ldots\) a text for the empty language.}
A \emph{text} for \(L\subseteq\Sigma^*\) is an infinite sequence over
\(L\cup\{\#\}\) in which every element of \(L\) appears at least once.
Pause symbols may occur arbitrarily, and \(\#\) is not an end-of-text marker.

The \emph{content} of a finite prefix is the finite set
\(K\subseteq L\) of positive examples from \(\Sigma^*\) that occur in that
prefix.  Thus neither the positions nor the number of pause symbols, nor repeated
occurrences of the same positive example, affect the content.
A learner is a computable function mapping each finite prefix to a hypothesis
grammar.

In this paper, a learner is said to identify a class \(\mathcal C\) in the
limit from positive data if, for every \(L\in\mathcal C\) and every text for
\(L\), there exists \(n_0\) such that every hypothesis produced from stage
\(n_0\) onward has language \(L\).  We call this \emph{behaviorally correct
identification from text} (TxtBc) \cite{CaseSmith1983}.  Thus eventual
semantic correctness of every hypothesis is required, but syntactic
stabilization of the hypothesis grammar is not.  The stronger Gold-style
explanatory (EX) criterion additionally requires the hypothesis grammar itself
to converge to a stable grammar.  We prove the semantic result with a
set-driven learner and recover explanatory convergence by a conservative
wrapper in Corollary~\ref{cor:gold-explanatory}.
\end{definition}

\begin{definition}[Set-driven learner and characteristic sample]
A learner \(\mathcal A\) is \emph{set-driven} if its output hypothesis is
independent of presentation order and multiplicity and depends only on the
finite set \(K\) of positive examples observed so far.  Write
\(\mathcal A(K)\) for the hypothesis on observation set \(K\).

For a target language \(L\), a finite set \(S\subseteq L\) is a
\emph{characteristic sample for \(\mathcal A\)} if, for every finite set
\(K\), \(S\subseteq K\subseteq L\) implies \(L(\mathcal A(K))=L\).
The case \(S=\emptyset\) is allowed.

The canonical learner constructed below is set-driven.  All learner-uniform
characteristic-sample lower bounds proved in this paper are stated within this
set-driven model.%
\footnote{Accordingly, this paper makes no characteristic-sample lower-bound
claim for arbitrary learners that may depend on presentation order or
multiplicity.}

The presentation-relative characteristic samples constructed below are
concrete characteristic samples for this canonical learner.  The target
presentation is used only to establish the existence and size of such a finite
characteristic sample; it is not supplied as input to the learner.
\end{definition}

\begin{definition}[Positive sample size]
For finite \(K\subseteq\Sigma^*\), define
\(\|K\|_+:=\sum_{w\in K}\max(1,|w|)\).%
\footnote{The factor \(\max(1,|w|)\) prevents an observed empty word from
making a degenerate zero contribution to sample size.}
For a tuple \(\tuple u=(u_1,\ldots,u_d)\), write
\(\|\tuple u\|_1:=\sum_{i=1}^d |u_i|\)
for its total component length.
\end{definition}

The general MCFG theorem below allows \(\eps\)-generating rank-zero rules.

\begin{definition}[Polynomial time and data]
\label{def:poly-time-data}
Following the characteristic-sample viewpoint of de~la~Higuera
\cite{delaHiguera1997,delaHiguera2010}, let \(\mathcal G\) be a class of
presentations and let
\(\mathcal L(\mathcal G):=\{L(G)\mid G\in\mathcal G\}\).
For fixed external parameters, we say that \(\mathcal G\) is
\emph{identified in polynomial time and data} if there exist polynomials
\(p\) and \(q\) such that, for every \(G\in\mathcal G\), there is a
characteristic sample \(S_G\subseteq L(G)\) for \(L(G)\) satisfying
\(\|S_G\|_+\le p(|G|)\),
and, for every finite sample \(K\subseteq\Sigma^*\), the learner's hypothesis
on \(K\) is constructible in time at most \(q(\|K\|_+)\), including output
size.  Thus the polynomial-data requirement is presentation-relative: the bound is
measured against the size of each witnessing presentation rather than a
canonical or minimal language description.  For the working binary
presentation classes used in the quantitative sections, we use the following
fixed binary encoding.  If \(P_{\mathrm{bin}}\) is the set of binary rules and
\(\lVert\rho\rVert_{\mathrm{temp}}\) is the total number of terminal and
variable symbols in the template of \(\rho\), then for fixed fan-out
\[
  |G|=\Theta\!\left(
    |V|+|P|\log\max\{2,|V|\}
    +\sum_{\rho\in P_{\mathrm{bin}}}\lVert\rho\rVert_{\mathrm{temp}}
  \right).
\]
On each fixed \((f,h)\)-slice, the alphabet and the explicit finite observation
are external parameters and are not included in \(|G|\).
\end{definition}

\begin{remark}[Slicewise polynomiality]
\label{rem:slicewise-poly}
Unless a later theorem makes additional parameter dependence explicit,
polynomial bounds are understood slicewise in the fixed external parameters.
In particular, Theorem~\ref{thm:poly} fixes \(f\), a branching cap \(b\), and \(h\), and gives the
construction bound \(\|K\|_+^{O(fb)}\); no bound polynomial uniformly in \(f\) or \(b\)
is claimed.  Theorem~\ref{thm:single-spine-polynomial-exposure} later makes
the dependence on the size of a supplied observation explicit, and
Theorem~\ref{thm:bounded-union-single-spine-poly} treats each fixed
observation-size bound \(k\) as a separate slice.
\end{remark}

\subsection{Finite monoid typing}

\begin{definition}[Explicit finite homomorphism]
\label{def-monoid-morphism}
A monoid homomorphism \(h\colon\Sigma^*\to M\) into a finite monoid \(M\) is
\emph{explicit} if the multiplication table of \(M\) and the values \(h(a)\) for
\(a\in\Sigma\) are given.  For a tuple \(\tuple w=(w_1,\ldots,w_d)\), write
\(h^{(d)}(\tuple w):=(h(w_1),\ldots,h(w_d))\in M^d\).
\end{definition}

In this paper, finite-monoid homomorphisms are used as finite-state auxiliary
observations for positive-data learning of MCFLs.  For the classical
correspondence between finite-monoid homomorphisms and finite-automaton
recognition of regular languages, see Pin~\cite{Pin1986}.

A finite observation \(h\) may arise, for example, when a DFA recognizing a
regular language \(R\supseteq L\) is available in advance.  The transition
morphism of this DFA yields an explicit homomorphism \(h:\Sigma^*\to M\): for
each word \(w\in\Sigma^*\), the element \(h(w)\) records the state
transformation induced by reading \(w\).

Thus \(h(u)=h(v)\) means that \(u\) and \(v\), when read from any DFA state,
lead to the same resulting state.  The learner does not infer \(R\) itself and
does not use \(R\) as a membership oracle for \(L\).  It uses only this finite
state information from the DFA as auxiliary information controlling when
tuples observed in positive data may be identified.

\subsection{Standard MCFG presentations and the binary working subfamily}

We use the standard tuple-generating definition of multiple context-free
grammars \cite{SekiEtAl1991}; see also \cite{Kallmeyer2010} for a general
overview of MCFGs and related mildly context-sensitive formalisms.
The explicit general-rank notation below is needed for the qualitative learner;
the more restrictive binary working form that follows is retained for the
quantitative single-spine analysis.

\begin{definition}[Standard linear MCFG presentation]
\label{def:standard-mcfg}
A \emph{multiple context-free grammar} is a tuple
\(G=(V,\Sigma,P,S,\mu)\), where \(V\) is a finite set of nonterminals,
\(\Sigma\) is a finite terminal alphabet, \(S\in V\) has \(\mu(S)=1\), and
\(\mu(A)\ge1\) is the fan-out of each \(A\in V\).  A rule of rank \(r\ge0\)
has the form
\[
  \rho:\quad
  A\to(\alpha_1,\ldots,\alpha_e)(B_1,\ldots,B_r),
\]
where \(e=\mu(A)\), \(d_j=\mu(B_j)\), and, with
\[
  X_j:=\{x_j^1,\ldots,x_j^{d_j}\},
  \qquad
  \Gamma_\rho:=\Sigma\uplus X_1\uplus\cdots\uplus X_r,
\]
each \(\alpha_i\in\Gamma_\rho^*\) and every variable in
\(X_1\uplus\cdots\uplus X_r\) occurs at most once in the complete tuple
\((\alpha_1,\ldots,\alpha_e)\).  The rule is \emph{nondeleting} if every such
variable occurs exactly once.

For compatible child tuples
\(\tuple u_j=(u_{j,1},\ldots,u_{j,d_j})\), simultaneous substitution of
\(u_{j,k}\) for \(x_j^k\) defines
\[
  \rho(\tuple u_1,\ldots,\tuple u_r)\in(\Sigma^*)^e.
\]
When \(r=0\), the rule is a constant rule and derives the tuple
\((\alpha_1,\ldots,\alpha_e)\).  The tuple languages \(L_A(G)\) form the least family, ordered componentwise by
inclusion, that is closed under all rule applications, and
\[
  L(G):=\{w\in\Sigma^*:(w)\in L_S(G)\}.
\]
An \(f\)-MCFG has \(\mu(A)\le f\) for every nonterminal, and an \(f\)-MCFL is
a language generated by an \(f\)-MCFG.  A grammar is \(b\)-branching if every
rule has rank at most \(b\).
\end{definition}

When two nonterminals have the same fan-out \(d\), we use
\[
  A\to B
\]
as shorthand for the unary identity-template rule
\(A\to(x^1,\ldots,x^d)(B)\).%
\footnote{This abbreviation denotes only the identity correspondence between
components.  For example, at fan-out two, \(A\to B\) means
\(A\to(x^1,x^2)(B)\); it does not denote a component-swapping rule such as
\(A\to(x^2,x^1)(B)\).}
Likewise, \(A\to\tuple c\) denotes a rank-zero
constant rule.  These abbreviations are used for typed start links and for the
semantic unary rules of the learner.

The standard formalism allows rank zero, nonidentity unary rules, arbitrary
finite rule rank, and deleting rules.  We use two known facts to separate the
semantic theorem from presentation artifacts.  First, every \(b\)-branching
\(f\)-MCFG has an equivalent \(b\)-branching \(f\)-MCFG whose rules are
nondeleting (indeed nondeleting and nonpermuting)
\cite{Kanazawa2019Ogden}.  Second, reducing arbitrary rule rank to two cannot
in general be done while preserving fan-out.%
\footnote{This reflects the fact that an MCFG cannot later split a string
component once it has been concatenated.  For example, consider fan-out-two
nonterminals \(A,B,C,D,E\) and the rank-four rule
\(A\to(x_B^1x_C^1x_D^1x_E^1,\,
 x_D^2x_B^2x_E^2x_C^2)(B,C,D,E)\).
For every pair of children, the corresponding variables are adjacent in one
parent component but are separated by variables of other children in the
other component.  Hence, if that pair is first combined into an auxiliary
nonterminal, pieces that must later be separated by material contributed by
other children have to be retained as distinct components, requiring at least
three components.  Every binary decomposition has a first pair of children to
combine, so this rule cannot be binarized while keeping fan-out two.}
LCFRS binarization may require
larger fan-out, and fan-out-preserving weakly equivalent binarizations do not
always exist \cite{GomezRodriguezEtAl2009}.  Thus the arbitrary-rank witness
construction below is not replaceable by a harmless binary normal-form step
when the fan-out parameter \(f\) is fixed.

\begin{definition}[Productivity, reachability, and reduction]
\label{def:productive-reachable-reduced}
Let \(G=(V,\Sigma,P,S,\mu)\) be a standard MCFG.

A nonterminal \(A\in V\) is \emph{productive} if
\(L_A(G)\ne\emptyset\).

The \emph{nonterminal dependency graph} of \(G\) has vertex set \(V\) and,
for every rule
\(A\to(\alpha_1,\ldots,\alpha_e)(B_1,\ldots,B_k)\)
and every \(1\le j\le k\), an edge \(A\to B_j\).
A nonterminal \(A\) is \emph{reachable} if there is a directed path from
\(S\) to \(A\) in this graph.  Paths of length zero are allowed, so the start
symbol \(S\) is always reachable.

The grammar \(G\) is \emph{reduced} if every nonterminal is both reachable
and productive.
\end{definition}

\paragraph{Binary working subfamily.}
The following restricted form is used for the compact examples, the
single-spine notion, and the polynomial-data analysis.  It is no longer part
of the definition of the qualitative target class.

\begin{definition}[Working binary linear nondeleting MCFG]
\label{def:working-mcfg}
A \emph{working binary linear nondeleting multiple context-free grammar} is a
tuple \(G=(V,\Sigma,P,S,\mu)\), where \(V\) is a finite nonterminal set,
\(\Sigma\) is a finite terminal alphabet, \(S\in V\) is the start symbol, and
\(\mu\colon V\to\mathbb N_{>0}\) is the fan-out map with \(\mu(S)=1\).  Rules
have one of the following forms.
\begin{enumerate}[label=(\roman*),leftmargin=*]
\item A start rule \(S\to A\), where \(A\ne S\) and \(\mu(A)=1\).
\item A terminal rule \(A\to(a)\), where \(A\ne S\), \(a\in\Sigma\), and
\(\mu(A)=1\).
\item A binary rule \(\rho\colon A\to(\alpha_1,\ldots,\alpha_e)(B,C)\), where
\(A\ne S\), \(e=\mu(A)\), \(\mu(B)=d_B\), \(\mu(C)=d_C\), each \(\alpha_i\) is
a word over \(\Sigma\cup\{x^1,\ldots,x^{d_B},y^1,\ldots,y^{d_C}\}\), and each
variable \(x^1,\ldots,x^{d_B},y^1,\ldots,y^{d_C}\) occurs exactly once in the
whole tuple \((\alpha_1,\ldots,\alpha_e)\), so the rule is nondeleting.
\end{enumerate}
The presentation is \emph{start-separated}: the only rules with left-hand side
\(S\) are start rules, and \(S\) never occurs as a child on the right-hand side
of a binary rule.
\end{definition}

\paragraph{Rule rank versus fan-out.}
The \emph{rank} of a production is the number of nonterminal children on its
right-hand side, whereas the fan-out \(\mu(A)\) is the arity of tuples derived
from \(A\).  These notions are independent.%
\footnote{Fan-out affects the complexity of describing individual tuple
occurrences, whereas rule rank controls how many child tuple occurrences must
be combined simultaneously when reconstructing one rule.  The latter therefore
enters directly into the combinatorial complexity of hypothesis construction.}
For example, if \(\mu(A)=\mu(B)=3\), the rule
\(A\to(x^1a,x^2,bx^3)(B)\) has rank \(1\), while \(A\) has fan-out \(3\).
Conversely, if \(\mu(C)=\mu(D)=\mu(E)=\mu(F)=1\), then
\(C\to(x^1y^1z^1)(D,E,F)\) has parent fan-out \(1\) but rule rank \(3\).
Apart from the distinguished start rules, the working form defined in this
paper has one-letter rank-zero terminal rules and rank-two composition rules.
It has no rank-one nonidentity composition rules.  This restriction concerns
only the working subfamily: the general learner returns standard MCFG
hypotheses and may retain unary identity template rules directly.

In what follows, for a binary rule
\(\rho\colon A\to(\alpha_1,\ldots,\alpha_e)(B,C)\),
we call the tuple \((\alpha_1,\ldots,\alpha_e)\) the
\emph{template tuple} of \(\rho\).
The semantics of rule application, as well as the tuple languages
\(L_A(G)\) and the generated language \(L(G)\), are those of
Definition~\ref{def:standard-mcfg}.

\begin{remark}[Empty components, the empty word, and the empty language]
\label{rem:working-mcfg-empty}
For working binary linear nondeleting MCFGs, three different notions of
emptiness must be distinguished.

First, an individual component of a tuple derived by a nonterminal may be the
empty word \(\eps\).  Indeed, a component \(\alpha_i\) of a binary-rule
template may itself be empty.  For example, if
\(B\to(a)\), \(C\to(b)\), and
\(A\to(x^1y^1,\eps)(B,C)\), then \(A\) derives the tuple
\((ab,\eps)\).

Second, a word generated by the grammar can never itself be \(\eps\).  Every
successful derivation tree contains at least one terminal leaf, and each
terminal rule \(A\to(a)\) contributes one terminal letter.  Since every
binary rule is linear and nondeleting, every component received from its
children is preserved exactly once in the parent composition, so a terminal
contribution cannot disappear higher in the derivation.  Finally, a start
rule \(S\to A\) has \(\mu(A)=1\), and hence all terminal contributions in a
successful derivation are ultimately contained in this single start
component.  Therefore every word in \(L(G)\) is nonempty and
\(L(G)\subseteq\Sigma^+\).

Third, the generated language \(L(G)\) may nevertheless be the empty set.
Definition~\ref{def:working-mcfg} does not require the grammar to be
productive, so, for example, \(L(G)=\emptyset\) if there is no successful
terminal derivation from the start symbol.  Likewise, an individual
nonterminal \(A\) may satisfy \(L_A(G)=\emptyset\).  If \(G\) is reduced,
however, the start symbol is productive in particular, and therefore
\(L(G)\ne\emptyset\).
\end{remark}

\begin{remark}[Role of the working subfamily]
\label{rem:working-form-scope}
The working binary form is a presentation-level subfamily used only where a
compact binary presentation is quantitatively relevant.  In particular,
positive fan-out, one-letter rank-zero axioms, start separation, and the
absence of nonidentity unary rules are \emph{not} assumptions of the general
reconstruction theorem.  For a standard target MCFG, the proof of qualitative
learnability first uses the fan-out-preserving nondeleting normalization cited
above and then works directly with arbitrary finite rule rank.

For working presentations, individual tuple components may still be empty
because template components lie in \(\Gamma_\rho^*\).  The single-spine
condition introduced below remains only a sufficient presentation condition
for polynomial characteristic data; it is not a semantic restriction on
\(\Cmcf{f}{h}\).
\end{remark}

The following presentation-level notion will be used in the examples and in
the polynomial-data analysis of Section~\ref{sec:poly-time}.

\begin{definition}[Lexical nonterminal]
\label{def:lexical-nonterminal}
Let \(G\) be a working binary linear nondeleting MCFG.
A nonterminal \(A\in V\setminus\{S\}\) is \emph{lexical} if every rule with
left-hand side \(A\) is a terminal rule \(A\to(a)\) with \(a\in\Sigma\).
Equivalently, every tuple derived from a lexical nonterminal has arity \(1\)
and consists of a single terminal letter.
\end{definition}

\begin{definition}[Single-spine working MCFG]
\label{def:single-spine-working-mcfg}
A working binary linear nondeleting MCFG \(G\) is \emph{single-spine} if, for
every binary rule \(A\to\rho(B,C)\), at most one of the two child
\emph{occurrences} is labelled by a nonlexical nonterminal.
In particular, if a rule \(A\to\rho(B,B)\) exists for a nonlexical
nonterminal \(B\), then \(G\) is not single-spine.
Equivalently, in every derivation tree the nonlexical child occurrences form
at most one downward spine, with lexical leaves attached to that spine.
\end{definition}

\begin{proposition}[Inlining lexical children]
\label{prop:single-spine-unary-rank-collapse}
Every single-spine working binary linear nondeleting MCFG has an equivalent
standard MCFG with the same fan-out bound in which every nonstart production
has rank at most one.
\end{proposition}

\begin{proof}
Consider a binary rule \(A\to\rho(B,C)\).  By single-spine, at least one child
occurrence is lexical.  If, say, \(C\) is lexical, then for every terminal rule
\(C\to(a)\) substitute \(a\) directly for the unique variable contributed by
that child in the template of \(\rho\).  This produces a rank-one rule with
child \(B\) and the same parent tuple for every derivation using that lexical
choice.  The same argument applies when \(B\) is lexical.  If both children are lexical,
expand the finitely many pairs of terminal choices and obtain rank-zero
constant rules.  Repeating this finite inlining operation for every binary rule
does not change the fan-out.
\end{proof}

\begin{corollary}[Fan-out-one single-spine presentations are linear]
\label{cor:fanout-one-single-spine-linear}
In particular, every language generated by a fan-out-one single-spine working
presentation is a linear context-free language.
\end{corollary}

\begin{proof}
Apply Proposition~\ref{prop:single-spine-unary-rank-collapse}.  At fan-out one,
a rank-at-most-one MCFG is an ordinary linear context-free grammar after
writing its unary templates as productions \(A\to uBv\) and its rank-zero
templates as terminal productions.
\end{proof}

\begin{remark}
Thus single-spine is a strong presentation restriction, not merely a mild
binary normal-form convention.  Its role here is quantitative: it is one
sufficient condition for polynomial exposure.  At fan-out two it still covers
genuinely non-context-free examples such as \(L_3\) and \(L_{\times}\).
\end{remark}

\begin{remark}[Compatible binary LCFRS notation]
Already-binary LCFRS presentations satisfying the working restrictions of
Definition~\ref{def:working-mcfg} are identified below with the corresponding
MCFG tuple-template notation.
For example, for fan-out-two nonterminals \(A,B,C\), the LCFRS rule
\(A(x_1ay_1,x_2by_2)\to B(x_1,x_2)\,C(y_1,y_2)\)
is written in the MCFG notation of this paper as
\(A\to(x^1ay^1,x^2by^2)(B,C)\).
The former writes the placement of child components directly on the left-hand
side, while the latter records the same placement as a tuple template; the two
notations describe the same composition operation.
This is only a notation-level correspondence and does not assert that an
arbitrary LCFRS can be normalized into the present working form.
\end{remark}

\subsection{Sentence contexts and tuple distributions}

\begin{definition}[Sentence context]
\label{def:sentence-context}
A \emph{sentence context of arity \(d\)} is a word of the form
\(E=u_0\blank_{\sigma(1)}u_1\cdots\blank_{\sigma(d)}u_d\),
where \(\sigma\) is a permutation of \(\{1,\ldots,d\}\), each
\(u_i\in\Sigma^*\), and each named hole \(\blank_i\) occurs exactly once.
For \(\tuple w=(w_1,\ldots,w_d)\), write \(E[\tuple w]\) for the string
obtained by substituting \(w_i\) for \(\blank_i\).
\end{definition}

\begin{definition}[Tuple occurrence in a sample word]
\label{def:tuple-occurrence}
Let \(K\subseteq\Sigma^*\) and let \(w=a_1\cdots a_n\in K\).
Number the cut positions of \(w\) as
\(\cutpos{0}a_1\cutpos{1}a_2\cutpos{2}\cdots a_n\cutpos{n}\).
For \(0\le\ell\le r\le n\), write
\(w[\ell:r]:=a_{\ell+1}\cdots a_r\); in particular,
\(w[\ell:\ell]=\eps\).

Let \(d\ge1\).  An \emph{arity-\(d\) tuple occurrence} in \(w\) is data
\(\mathfrak{o}=(\sigma,(\ell_i,r_i)_{i=1}^d)\), where \(\sigma\) is a
permutation of \(\{1,\ldots,d\}\) and
\(0\le\ell_{\sigma(1)}\le r_{\sigma(1)}
\le\cdots\le\ell_{\sigma(d)}\le r_{\sigma(d)}\le n\).
Empty intervals \(\ell_i=r_i\) are allowed, and when several empty intervals
occupy the same cut position, \(\sigma\) also records their local order.

For such an occurrence, put
\(\tuple x=(w[\ell_1:r_1],\ldots,w[\ell_d:r_d])\).
Replacing the specified intervals by their corresponding named holes
\(\blank_i\) yields an arity-\(d\) sentence context \(E\) satisfying
\(E[\tuple x]=w\).  We identify the occurrence with the pair
\((E,\tuple x)\).

For example, let
\(w=\cutpos{0}a\cutpos{1}b\cutpos{2}c\cutpos{3}d\cutpos{4}\), let \(d=2\),
and take \((\ell_1,r_1)=(2,4)\) and \((\ell_2,r_2)=(0,1)\).
The second component occurs first from the left, so \(\sigma=(2,1)\);
the corresponding tuple is \((cd,a)\) and the corresponding sentence context
is \(\blank_2b\blank_1\).
\end{definition}

Template components in MCFG rules may be empty, which is why
Definition~\ref{def:tuple-occurrence} allows empty intervals.
For example, for \(w=ab\) with cut positions
\(\cutpos{0}a\cutpos{1}b\cutpos{2}\), the empty word at the middle cut is
represented by the interval \([1,1]\).
Thus, if both components of the arity-two tuple \((\eps,\eps)\) are observed
at this cut, then \((\ell_1,r_1)=(\ell_2,r_2)=(1,1)\).
Even in this case, \(\sigma=(1,2)\) and \(\sigma=(2,1)\) determine,
respectively, \(a\blank_1\blank_2b\) and \(a\blank_2\blank_1b\), and hence
are treated as distinct tuple occurrences.

\begin{definition}[Tuple distribution]
\label{def:tuple-distribution}
Let \(d\ge1\).  For \(L\subseteq\Sigma^*\) and
\(\tuple x\in(\Sigma^*)^d\), define
\[
  \D_L^{(d)}(\tuple x)
  :=
  \left\{
    E
    \;\middle|\;
    E\text{ is an arity-}d\text{ sentence context and }
    E[\tuple x]\in L
  \right\}.
\]
When the arity is clear, write \(\D_L(\tuple x)\).
\end{definition}

\begin{definition}[\((f,h)\)-tuple substitutability]
\label{def:tuple-substitutable}
Let \(f\ge1\) and let \(h\colon\Sigma^*\to M\) be an explicit finite monoid
homomorphism.  A language \(L\subseteq\Sigma^*\) is
\emph{\((f,h)\)-tuple-substitutable} if, for every \(1\le d\le f\) and all
\(\tuple x,\tuple y\in(\Sigma^*)^d\), the implication
\(h^{(d)}(\tuple x)=h^{(d)}(\tuple y)\) and
\(\D_L^{(d)}(\tuple x)\cap\D_L^{(d)}(\tuple y)\ne\emptyset\) imply
\(\D_L^{(d)}(\tuple x)=\D_L^{(d)}(\tuple y)\).
Recall from Definition~\ref{def-monoid-morphism} that
\(h^{(d)}(\tuple w)=(h(w_1),\ldots,h(w_d))\in M^d\).
\end{definition}

For \(d=1\), sentence contexts are ordinary two-sided contexts
\(u\blank_1v\), so Definition~\ref{def:tuple-substitutable} specializes to the
fixed-\(h\) two-sided substitutability condition for CFLs
\cite{kuriyamaCFG}.

\begin{definition}[The target class]
\label{def:class-cmcf}
Fix \(f\ge1\) and an explicit finite monoid homomorphism \(h\).  The class
\(\Cmcf{f}{h}\) consists of all languages \(L\subseteq\Sigma^*\) such that
\[
  L\in f\text{-}\mathrm{MCFL}
  \qquad\text{and}\qquad
  L\text{ is }(f,h)\text{-tuple-substitutable}.
\]
\end{definition}

For fixed \(f\) and \(h\), we call \(\Cmcf{f}{h}\) the
\emph{fixed-observation slice determined by \(h\)}.  The word ``slice''
emphasizes that the observation morphism is fixed externally and shared by
all targets in the class; it is not inferred separately from each positive
text.

\begin{remark}[Semantic nature of the target class]
\label{rem:semantic-class}
Membership in \(\Cmcf{f}{h}\) is a semantic promise imposed on the target
language.  The learner is defined on every finite positive sample, but it does
not receive an MCFG presentation of the target language and does not test this
promise from positive data.  The present paper neither asserts nor uses a
procedure that decides this promise from a given MCFG presentation.  Whether
\((f,h)\)-tuple substitutability is decidable from an arbitrary standard MCFG
presentation is outside the scope of the paper.  Related decision problems for
congruential properties have been studied in the context-free setting
\cite{Clark2017,KanazawaKappe2019}.
\end{remark}

\begin{definition}[Refinement of finite observation morphisms]
\label{def:h-refinement}
Let \(h\colon\Sigma^*\to M\) and \(h'\colon\Sigma^*\to M'\) be explicit finite
homomorphisms.  We say that \(h'\) \emph{refines} \(h\), and write
\(h\preceq h'\), if there is a monoid homomorphism
\(\pi\colon M'\to M\) such that \(h=\pi\circ h'\).
\end{definition}

\begin{proposition}[Monotonicity under refinement of the observation morphism]
\label{prop:h-refinement-monotonicity}
If \(h\preceq h'\) and \(L\) is \((f,h)\)-tuple-substitutable, then \(L\) is
\((f,h')\)-tuple-substitutable.  Consequently, for every pair of explicit
finite homomorphisms with \(h\preceq h'\),
\(\Cmcf{f}{h}\subseteq \Cmcf{f}{h'}\).
\end{proposition}

\begin{proof}
Let \(\pi\colon M'\to M\) satisfy \(h=\pi\circ h'\).  Suppose that
\(h'^{(d)}(\tuple x)=h'^{(d)}(\tuple y)\) and that \(\tuple x,\tuple y\) share an
accepting arity-\(d\) sentence context.  Applying \(\pi\) componentwise gives
\(h^{(d)}(\tuple x)=h^{(d)}(\tuple y)\).  The \((f,h)\)-tuple-substitutability of
\(L\) therefore gives \(\D_L^{(d)}(\tuple x)=\D_L^{(d)}(\tuple y)\).
This is precisely the \((f,h')\)-tuple-substitutability condition.
Membership in \(f\text{-}\mathrm{MCFL}\) is unchanged, so the class
inclusion follows.
\end{proof}

\begin{example}[Refinement distinguishes grammatical phases]
\label{ex:refinement-grammar-intuition}
Let \(\Sigma=\{a,b\}\) and consider
\(L_{\mathrm{eq}}:=\{a^n b^n:n\ge1\}\).

First, the trivial observation
\(\mathbf 1\colon\Sigma^*\to\{e\}\)
maps every string to the same observation value.

Now let
\(M_{\mathrm{phase}}=(\{0,1\},\lor,0)\)
and define
\(h_{\mathrm{phase}}\colon\Sigma^*\to M_{\mathrm{phase}}\)
to return \(0\) if no \(b\) has yet occurred and \(1\) once a \(b\) has
occurred.  Equivalently,
\(h_{\mathrm{phase}}(w)=0\) for \(w\in a^*\), and
\(h_{\mathrm{phase}}(w)=1\) otherwise.
If
\(\pi\colon M_{\mathrm{phase}}\to\{e\}\)
is the unique homomorphism to the one-element monoid, then
\(\mathbf 1=\pi\circ h_{\mathrm{phase}}\), so
\(\mathbf 1\preceq h_{\mathrm{phase}}\).

Consider the substrings \(x=a\) and \(y=a^2b\).
For the sentence context \(E=\blank_1 b\),
\(E[x]=ab\in L_{\mathrm{eq}}\) and
\(E[y]=a^2b^2\in L_{\mathrm{eq}}\).
Positive data alone therefore make the two substrings look like plausible
candidates for identification.

They are not fully interchangeable.  For
\(F=\blank_1 abb\),
\(F[x]=a^2b^2\in L_{\mathrm{eq}}\), whereas
\(F[y]=a^2babb\notin L_{\mathrm{eq}}\).
Thus identifying \(x\) and \(y\) would be unsafe.

The trivial observation \(\mathbf 1\) cannot distinguish them, but
\(h_{\mathrm{phase}}(a)=0\) and
\(h_{\mathrm{phase}}(a^2b)=1\).
The refined observation therefore distinguishes a substring that is still in
the \(a\)-phase from one that has already entered the \(b\)-phase, preventing
this unsafe grammatical identification.

In this example, refinement does more than increase the number of observation
values: it retains finite-state information about the stage of the generative
process that is difficult to recover from positive data alone.
The phenomenon already occurs at fan-out one, so the role of finite
observation is not specific to multiple context-free grammars.
\end{example}

\begin{definition}[Fixed-\(h\) distributional equivalence]
\label{def:distributional-equivalence}
For \(\tuple{u},\tuple{x}\in(\Sigma^*)^d\), write
\(\tuple{u}\equiv_L^d\tuple{x}\) if
\(h^{(d)}(\tuple{u})=h^{(d)}(\tuple{x})\) and
\(\D_L^{(d)}(\tuple{u})=\D_L^{(d)}(\tuple{x})\).
\end{definition}

\begin{remark}[Distributional equivalence and concrete witnesses]
\label{rem:distributional-equivalence-witness}
The relation \(\equiv_L^d\) is semantic.  Its validity by itself does not
require the existence of a concrete accepting context connecting the two
tuples, nor of a witness appearing in the positive sample.  Indeed, it may
happen that
\(\D_L^{(d)}(\tuple u)=\D_L^{(d)}(\tuple x)=\emptyset\); if the two tuples
also have the same \(h\)-type, then
\(\tuple u\equiv_L^d\tuple x\) still holds.

By contrast, the shared accepting context used in
Lemma~\ref{lem:shared-context}, as well as the concrete tuple occurrences and
composition witnesses enumerated later by the learner, are finite pieces of
evidence observable from positive data.  We therefore distinguish the
semantic relation \(\equiv_L^d\) from the concrete witnesses used to justify
that relation from positive data.
\end{remark}

\begin{lemma}[Shared-context substitutability]
\label{lem:shared-context}
Let \(L\) be \((f,h)\)-tuple-substitutable.  If \(d\le f\),
\(h^{(d)}(\tuple x)=h^{(d)}(\tuple y)\), and a concrete arity-\(d\) context
\(E\) satisfies \(E[\tuple x],E[\tuple y]\in L\), then
\(\tuple x\equiv_L^d\tuple y\).
\end{lemma}

\begin{proof}
The context \(E\) belongs to both tuple distributions, so
Definition~\ref{def:tuple-substitutable} gives equality of the distributions;
the type equality is assumed.
\end{proof}

\section{Examples}
\label{sec:running-examples}

Before the abstract reconstruction proof, we record concrete grammars showing
what the fixed observation morphism captures in familiar block-synchronization
languages.  A fixed regular envelope supplies only coarse information such as
the order of the letter blocks.  Equalities between block lengths are not
encoded by the envelope; they are recovered by the learner from the
sentence-context distributions of tuples observed in positive examples.  We
also verify the semantic promise for these examples, so that they really belong
to the relevant fixed-observation slices.

As in the examples below, MCFLs defined by numerical relations among letter
blocks provide basic test cases for the expressive power of the formalism.
For comparisons of the numbers of letters in consecutive blocks and the
resulting separations in the MCFL fan-out hierarchy, see
\cite{LehnerLindorfer2021}.

\begin{proposition}[Regular control and an abelian balance constraint]
\label{prop:regular-abelian-balance}
Let \(R\subseteq\Sigma^*\) be a regular language, represented by a finite
monoid \(M\), a homomorphism \(h_R\colon\Sigma^*\to M\), and a subset
\(P\subseteq M\) such that \(R=h_R^{-1}(P)\).
Let \(A\) be an abelian group written additively,
let \(\varphi\colon\Sigma^*\to A\) be a monoid homomorphism, and let
\(H\le A\) be a subgroup.
Then
\(L:=R\cap\varphi^{-1}(H)\)
is \((f,h_R)\)-tuple-substitutable for every \(f\ge1\).
\end{proposition}

\begin{proof}
Fix \(1\le d\le f\), and let
\(\tuple{x}=(x_1,\ldots,x_d)\) and
\(\tuple{y}=(y_1,\ldots,y_d)\) satisfy
\(h_R^{(d)}(\tuple{x})=h_R^{(d)}(\tuple{y})\) and
\(\D_L^{(d)}(\tuple{x})\cap\D_L^{(d)}(\tuple{y})\ne\emptyset\).
It suffices to show
\(\D_L^{(d)}(\tuple{x})=\D_L^{(d)}(\tuple{y})\).

Fix
\(E\in\D_L^{(d)}(\tuple{x})\cap\D_L^{(d)}(\tuple{y})\).
Then \(E[\tuple{x}],E[\tuple{y}]\in L\), and in particular
\(\varphi(E[\tuple{x}]),\varphi(E[\tuple{y}])\in H\).

First consider the regular control.
Let
\(F=u_0\blank_{\sigma(1)}u_1\cdots
\blank_{\sigma(d)}u_d\)
be an arbitrary arity-\(d\) sentence context.
Since
\(h_R^{(d)}(\tuple{x})=h_R^{(d)}(\tuple{y})\),
we have \(h_R(x_i)=h_R(y_i)\) for every \(i\), and hence
\(h_R(F[\tuple{x}])=h_R(F[\tuple{y}])\).
Because \(R=h_R^{-1}(P)\),
\(F[\tuple{x}]\in R\) if and only if \(F[\tuple{y}]\in R\).

Next consider the abelian balance constraint.
Put
\(\Delta:=\sum_{i=1}^d(\varphi(y_i)-\varphi(x_i))\in A\).
Since \(A\) is abelian, the order of the holes is irrelevant, and for every
arity-\(d\) sentence context \(F\),
\(\varphi(F[\tuple{y}])-\varphi(F[\tuple{x}])=\Delta\).
The fixed terminal parts of \(F\) cancel, leaving only the sum of the
componentwise differences.

Taking \(F=E\) gives
\(\Delta=\varphi(E[\tuple{y}])-\varphi(E[\tuple{x}])\).
Both terms lie in \(H\), so \(\Delta\in H\).
Hence for every \(F\),
\(\varphi(F[\tuple{y}])=\varphi(F[\tuple{x}])+\Delta\), and therefore
\(F[\tuple{x}]\in\varphi^{-1}(H)\) if and only if
\(F[\tuple{y}]\in\varphi^{-1}(H)\).

Combining this with \(L=R\cap\varphi^{-1}(H)\), we obtain, for every
arity-\(d\) sentence context \(F\),
\(F[\tuple{x}]\in L\) if and only if \(F[\tuple{y}]\in L\).
Thus
\(\D_L^{(d)}(\tuple{x})=\D_L^{(d)}(\tuple{y})\).

Since \(d\le f\) was arbitrary, \(L\) is
\((f,h_R)\)-tuple-substitutable.
\end{proof}

\begin{example}[A non-block example: MIX]
\label{ex:mix}
Let \(\Sigma:=\{a,b,c\}\) and define
\(\mathrm{MIX}:=\{w\in\Sigma^*:|w|_a=|w|_b=|w|_c\}\).
Salvati proved that \(\mathrm{MIX}\) is a \(2\)-MCFL
\cite{Salvati2015}.

Let \(A:=\mathbb Z^2\) be the additive group and define
\(\varphi\colon\Sigma^*\to\mathbb Z^2\) by
\(\varphi(w):=(|w|_a-|w|_b,\ |w|_b-|w|_c)\).
By additivity of Parikh counts, \(\varphi\) is a monoid homomorphism from
\((\Sigma^*,\cdot,\eps)\) to \((\mathbb Z^2,+,(0,0))\).
If \(H:=\{(0,0)\}\le\mathbb Z^2\), then for every \(w\in\Sigma^*\),
\(w\in\mathrm{MIX}\) if and only if \(\varphi(w)\in H\).  Hence
\(\mathrm{MIX}=\varphi^{-1}(H)\).

Next let \(\mathbf 1\colon\Sigma^*\to\{e\}\) be the unique homomorphism to
the one-element monoid.  With \(R:=\Sigma^*\), we have
\(R=\mathbf 1^{-1}(\{e\})\).  Proposition~\ref{prop:regular-abelian-balance},
applied with \(h_R=\mathbf 1\), \(A=\mathbb Z^2\), and
\(H=\{(0,0)\}\), therefore shows that \(\mathrm{MIX}\) is
\((f,\mathbf 1)\)-tuple-substitutable for every \(f\ge1\).

Since \(\mathrm{MIX}\) is a \(2\)-MCFL, it is an \(f\)-MCFL for every
\(f\ge2\).  Hence Definition~\ref{def:class-cmcf} gives
\(\mathrm{MIX}\in\Cmcf{f}{\mathbf 1}\) for every \(f\ge2\).
Finally, the one-element monoid has the minimum possible cardinality of a
finite monoid, so Definition~\ref{def:intrinsic-observation-size} gives
\(\obs_f(\mathrm{MIX})\le1\).  Conversely every finite monoid image has
cardinality at least \(1\), so \(\obs_f(\mathrm{MIX})\ge1\).  Therefore
\(\obs_f(\mathrm{MIX})=1\) for every \(f\ge2\).
\end{example}

\begin{remark}[MIX is not single-spine at fan-out two]
Although MIX is a \(2\)-MCFL \cite{Salvati2015}, it does not admit a
fan-out-two single-spine working presentation.
Indeed, Proposition~\ref{prop:single-spine-unary-rank-collapse} would turn any
such presentation into a non-branching \(2\)-MCFG.
Non-branching MCFGs reduce to the well-nested case, and well-nested
\(2\)-MCFLs coincide with tree-adjoining languages.
Kanazawa and Salvati showed, however, that MIX is not a tree-adjoining
language \cite{KanazawaSalvati2012}.
\end{remark}

\begin{proposition}[The copy language has no finite witnessing observation]
\label{prop:copy-infinite-observation}
Let \(\Sigma:=\{a,b\}\) and define
\(\mathrm{COPY}:=\{ww:w\in\Sigma^+\}\).
Then \(\mathrm{COPY}\in2\text{-}\mathrm{MCFL}\) but
\(\obs_2(\mathrm{COPY})=\infty\).  Consequently,
\(\Cmcf{2}{\bullet}\subsetneq2\text{-}\mathrm{MCFL}\).
\end{proposition}

\begin{proof}
Intuitively, the obstruction is that \(L_3\) only requires finite-state
information about the current block, whereas COPY requires the first half
\(w\) itself to remain distinguishable.
Every finite observation identifies some distinct strings \(u\ne v\), so the
exact identity information required by COPY cannot be compressed into finitely
many observation states.

We first show that \(\mathrm{COPY}\in2\text{-}\mathrm{MCFL}\).
Use a nonterminal \(A\) of fan-out \(2\) and a start symbol \(S\) of fan-out
\(1\).  Give \(A\) the rules
\(A\to(a,a)\), \(A\to(b,b)\),
\(A\to(a x^1,a x^2)(A)\), and
\(A\to(b x^1,b x^2)(A)\), and use the start rule
\(S\to(x^1x^2)(A)\).
Then \(L_A=\{(w,w):w\in\Sigma^+\}\).  Indeed, the two base rules derive
\((a,a)\) and \((b,b)\), while each recursive rule maps an already derived
\((w,w)\) to either \((aw,aw)\) or \((bw,bw)\), and no other form can be
derived.  Thus the start rule generates exactly
\(\{ww:w\in\Sigma^+\}\), proving that
\(\mathrm{COPY}\in2\text{-}\mathrm{MCFL}\).

Now fix an arbitrary finite monoid \(M\) and an arbitrary monoid homomorphism
\(h\colon\Sigma^*\to M\).  We show that \(\mathrm{COPY}\) is not
\((2,h)\)-tuple-substitutable.  Choose \(n\ge1\) with \(2^n>|M|\).
Since there are exactly \(2^n\) words of length \(n\), the pigeonhole
principle yields distinct \(u,v\in\Sigma^n\) with \(h(u)=h(v)\).

Let \(\tuple x:=(u,u)\) and \(\tuple y:=(v,v)\).  Then
\(h^{(2)}(\tuple x)=h^{(2)}(\tuple y)\).  Moreover, for the arity-two named
context \(E:=\blank_1\blank_2\), we have
\(E[\tuple x]=uu\in\mathrm{COPY}\) and
\(E[\tuple y]=vv\in\mathrm{COPY}\).  Hence
\(E\in\D_{\mathrm{COPY}}^{(2)}(\tuple x)\cap
\D_{\mathrm{COPY}}^{(2)}(\tuple y)\).

The two tuple distributions are nevertheless different.  Let
\(F:=\blank_1\blank_2uu\).  Then
\(F[\tuple x]=uuuu=(uu)(uu)\in\mathrm{COPY}\).  On the other hand,
\(F[\tuple y]=vvuu\).  Suppose for contradiction that
\(vvuu\in\mathrm{COPY}\).  Then \(vvuu=zz\) for some \(z\in\Sigma^+\).
Since \(|vvuu|=4n\), we have \(|z|=2n\).  Thus the first and second halves
of \(vvuu\), each of length \(2n\), must be equal, so \(vv=uu\).
Because \(|u|=|v|=n\), comparing the first \(n\) symbols gives \(u=v\), a
contradiction.  Therefore \(F[\tuple y]\notin\mathrm{COPY}\).
Consequently
\(\D_{\mathrm{COPY}}^{(2)}(\tuple x)\neq
\D_{\mathrm{COPY}}^{(2)}(\tuple y)\).

Thus \(\tuple x\) and \(\tuple y\) have the same componentwise \(h\)-type
and share an accepting context, but have different tuple distributions.
By Definition~\ref{def:tuple-substitutable}, \(\mathrm{COPY}\) is not
\((2,h)\)-tuple-substitutable.  Since the finite-monoid homomorphism \(h\)
was arbitrary, no finite observation witnesses \((2,h)\)-tuple
substitutability for \(\mathrm{COPY}\), and hence
\(\obs_2(\mathrm{COPY})=\infty\).

Finally, Definition~\ref{def:no-advice-union} gives
\(\Cmcf{2}{\bullet}\subseteq2\text{-}\mathrm{MCFL}\).  The language
\(\mathrm{COPY}\) belongs to \(2\text{-}\mathrm{MCFL}\) but, by the result
just proved, does not belong to \(\Cmcf{2}{\bullet}\).  Therefore
\(\Cmcf{2}{\bullet}\subsetneq2\text{-}\mathrm{MCFL}\).
\end{proof}

For the positive block envelope
\(P_m:=a_1^+a_2^+\cdots a_m^+\),
let \(D_m^+\) be the complete DFA with state set
\(\{0,1,\ldots,m,\mathsf{sink}\}\),
initial state \(0\), and unique accepting state \(m\).

For each \(1\le i<m\), use the transitions
\[
  i\xrightarrow{a_i}i,
  \qquad
  i\xrightarrow{a_{i+1}}i+1,
\]
together with
\(0\xrightarrow{a_1}1\) and
\(m\xrightarrow{a_m}m\).
All transitions not listed above go to \(\mathsf{sink}\), and
\(\mathsf{sink}\) loops on every letter.

Let
\(h_m^+:\Sigma_m^*\to M_m^+\)
be the transition morphism of \(D_m^+\).
Then \(D_m^+\) recognizes exactly \(P_m\).

\begin{example}[The three-block agreement language]
\label{ex:l3-language}
Let
\(L_3:=\{a^n b^n c^n\mid n\ge1\}\) and
\(P_3:=a^+b^+c^+\), and put \(h_{P_3}:=h_3^+\).
This morphism records the coarse block zone of a component but not the equality
of the three block lengths.
\end{example}

\begin{proposition}[A single-spine working presentation for \(L_3\)]
\label{prop:l3-running-example}
The language \(L_3\) has a reduced single-spine working binary linear
nondeleting MCFG presentation of fan-out two.
\end{proposition}

\begin{proof}
Let \(V=\{S,T,A,A_a,A_b,A_c\}\), let \(\mu(A)=2\), and let all other
nonterminals have fan-out \(1\).
Use
\[
  S\to T,\quad
  A_a\to(a),\quad
  A_b\to(b),\quad
  A_c\to(c),\quad
  A\to(x^1,y^1)(A_a,A_b),
\]
\[
  A\to(ax^1,bx^2y^1)(A,A_c),\qquad
  T\to(x^1x^2y^1)(A,A_c).
\]

A direct induction gives
\[
  L_A=\{(a^n,b^nc^{n-1})\mid n\ge1\}.
\]
Hence the top rule
\(T\to(x^1x^2y^1)(A,A_c)\),
together with \(S\to T\), generates exactly \(L_3\).

Moreover, \(A_a,A_b,A_c\) are lexical.
In the rule
\(A\to(x^1,y^1)(A_a,A_b)\)
both children are lexical, while each of the other two binary rules has only
one nonlexical child, namely \(A\).
Thus every binary rule has at most one nonlexical child occurrence, and the
presentation is single-spine in the sense of
Definition~\ref{def:single-spine-working-mcfg}.

Finally, every nonterminal is reachable and productive, and every rule is
linear and nondeleting.
Therefore the grammar is a reduced single-spine working binary linear
nondeleting MCFG of fan-out two.
\end{proof}

\begin{corollary}[The three-block example belongs to the fixed-observation class]
\label{cor:l3-in-fixed-observation-class}
For every \(f\ge2\), \(L_3\in\Cmcf{f}{h_{P_3}}\).
\end{corollary}

\begin{proof}
The presentation is given by Proposition~\ref{prop:l3-running-example}.
For the regular control \(P_3=a^+b^+c^+\), define
\(\varphi(w):=(|w|_a-|w|_b,\ |w|_b-|w|_c)\in\mathbb Z^2\).
Then \(L_3=P_3\cap\varphi^{-1}(\{(0,0)\})\), so
Proposition~\ref{prop:regular-abelian-balance} gives
\((f,h_{P_3})\)-tuple substitutability.
\end{proof}

\begin{example}[The cross-serial two-parameter language]
\label{ex:cross-serial-language}
Let
\(L_{\times}:=\{a^n b^m c^n d^m\mid n,m\ge1\}\) and
\(Q:=a^+b^+c^+d^+\), and put \(h_Q:=h_4^+\).
The language has two independent agreement parameters, although the working
presentation below realizes them along one nonlexical derivation spine.
\end{example}

\begin{proposition}[A single-spine working presentation for \(L_{\times}\)]
\label{prop:cross-serial-running-example}
The language \(L_{\times}\) has a reduced single-spine working binary linear
nondeleting MCFG presentation of fan-out two, in the sense of
Definition~\ref{def:single-spine-working-mcfg}.
\end{proposition}

\begin{proof}
Let \(V=\{S,T,A,A_a,A_b,A_d\}\), let \(\mu(A)=2\), and let all other
nonterminals have fan-out \(1\).
Use the lexical rules
\[
  A_a\to(a),\qquad A_b\to(b),\qquad A_d\to(d),
\]
together with
\[
  A\to(x^1y^1,c)(A_a,A_b),\qquad
  A\to(y^1x^1,cx^2)(A,A_a),\qquad
  A\to(x^1b,x^2y^1)(A,A_d),
\]
and the top rules
\[
  T\to(x^1x^2y^1)(A,A_d),\qquad S\to T.
\]
A direct induction gives
\[
  L_A=\{(a^nb^m,c^nd^{m-1})\mid n,m\ge1\}.
\]
The first rule gives \((ab,c)\), the second increments \(n\), and the third
increments \(m\).  The top rules therefore generate exactly
\(L_{\times}\).

The nonterminals \(A_a,A_b,A_d\) are lexical.
Every binary rule has at most one nonlexical child, namely \(A\), so the
presentation is single-spine.
It is also reduced, linear, nondeleting, and of fan-out two.
\end{proof}

\begin{corollary}[The cross-serial example belongs to the fixed-observation class]
\label{cor:cross-serial-in-fixed-observation-class}
For every \(f\ge2\), \(L_{\times}\in\Cmcf{f}{h_Q}\).
\end{corollary}

\begin{proof}
The presentation is Proposition~\ref{prop:cross-serial-running-example}.
For \(Q=a^+b^+c^+d^+\), define
\(\varphi(w):=(|w|_a-|w|_c,\ |w|_b-|w|_d)\in\mathbb Z^2\).
Then \(L_{\times}=Q\cap\varphi^{-1}(\{(0,0)\})\), and
Proposition~\ref{prop:regular-abelian-balance} applies.
\end{proof}

\section{The Canonical Learner}
\label{sec:learner}

The learner uses only the fixed fan-out bound \(f\), the supplied observation
\(h\), and the finite positive sample available at the current stage.  A target
MCFG is used only on the proof side to establish the existence of a finite
characteristic sample.  The target presentation is never supplied to the
learner.

For the qualitative theorem, let \(G\) be any finite nondeleting standard MCFG
of fan-out at most \(f\).  In the final exact-reconstruction theorem such a
presentation is obtained, without changing fan-out or branching rank, from an
arbitrary standard MCFG presentation by the normalization recalled in
Section~\ref{sec:prelim}.

The proof architecture is worth making explicit.  Starting from an arbitrary
standard \(f\)-MCFG presentation of the target language, we first pass on the
proof side to an equivalent nondeleting presentation without increasing
fan-out or branching rank.  We then refine nonterminals by their finite
observation output types and eliminate genuinely silent child occurrences by
compression.  The resulting positive-rank rules are exposed directly, without
binarization.  Finally, although their ranks are not bounded uniformly over
the target class, every rule needed for reconstruction has a positive sample
exposure whose length bounds that particular rank.  Thus the general learner
can enumerate all required witnesses using only its sample-dependent rank cap.

\subsection{Output-Type Refinement}
\label{sec:refinement}

\begin{definition}[Template output type]
\label{def:output-map}
Let
\(\rho:A\to(\alpha_1,\ldots,\alpha_e)(B_1,\ldots,B_r)\)
be a nondeleting rule, put \(e=\mu(A)\) and
\(d_j:=\mu(B_j)\), and choose
\(\mathbf q_j\in M^{d_j}\) for each \(1\le j\le r\).

On the template alphabet
\(
  \Gamma_\rho
  :=\Sigma\uplus\biguplus_{j=1}^r
  \{x_j^1,\ldots,x_j^{d_j}\}
\),
define
\(\theta_{\mathbf q_1,\ldots,\mathbf q_r}\colon\Gamma_\rho\to M\)
by
\(\theta_{\mathbf q_1,\ldots,\mathbf q_r}(a):=h(a)\)
for \(a\in\Sigma\), and
\(\theta_{\mathbf q_1,\ldots,\mathbf q_r}(x_j^k):=(\mathbf q_j)_k\).

Let
\(
  \widehat\theta_{\mathbf q_1,\ldots,\mathbf q_r}
  \colon\Gamma_\rho^*\to M
\)
be its unique extension to a monoid homomorphism.
Define
\[
  \out_\rho(\mathbf q_1,\ldots,\mathbf q_r)
  :=
  \bigl(
    \widehat\theta_{\mathbf q_1,\ldots,\mathbf q_r}(\alpha_1),
    \ldots,
    \widehat\theta_{\mathbf q_1,\ldots,\mathbf q_r}(\alpha_e)
  \bigr)
  \in M^e.
\]
For rank zero, \(\out_\rho\) is simply
\((h(\alpha_1),\ldots,h(\alpha_e))\), the componentwise \(h\)-type of the
constant template.
\end{definition}

\begin{definition}[Output-type refinement]
\label{def:output-refinement}
Given a finite nondeleting MCFG \(G=(V,\Sigma,P,S,\mu)\) and a finite
observation \(h\), the \emph{full output-type refinement} \(G^h\) has a fresh
start symbol \(\widetilde S\) and typed nonterminals
\(A_{\mathbf p}\) for every \(A\in V\) and
\(\mathbf p\in M^{\mu(A)}\).
Set
\(\mu(\widetilde S):=1\) and
\(\mu(A_{\mathbf p}):=\mu(A)\).

For each \(p\in M\), add the start rule
\(\widetilde S\to S_{(p)}\).
For each rank-zero rule \(A\to\tuple c\), add
\(A_{h^{(\mu(A))}(\tuple c)}\to\tuple c\).

For every positive-rank rule
\(\rho:A\to(\alpha_1,\ldots,\alpha_e)(B_1,\ldots,B_r)\)
and every choice
\(\mathbf q_j\in M^{\mu(B_j)}\) (\(1\le j\le r\)), add
\[
  A_{\out_\rho(\mathbf q_1,\ldots,\mathbf q_r)}
  \to
  \rho(B_{1,\mathbf q_1},\ldots,B_{r,\mathbf q_r}).
\]
Here \(B_{j,\mathbf q_j}\) denotes the typed copy of \(B_j\) indexed by
\(\mathbf q_j\).

Finally, the \emph{trimmed refinement} \(\widetilde G_0\) is obtained by
retaining only the typed nonterminals and typed rules that actually occur in
some successful derivation from \(\widetilde S\).
\end{definition}

\begin{proposition}[Output-type invariants]
\label{prop:output-invariants}
For \(G^h\) and \(\widetilde G_0\):
\begin{enumerate}[label=(\roman*),leftmargin=*]
\item if \(A_{\mathbf p}\Rightarrow^*\tuple u\), then
\(h^{(\mu(A))}(\tuple u)=\mathbf p\), and erasing the type indices gives a
valid \(G\)-derivation of \(\tuple u\);
\item every fixed \(G\)-derivation tree lifts uniquely to \(G^h\) by labeling
each node with the componentwise \(h\)-type of the tuple derived at that node;
\item \(L(\widetilde G_0)=L(G^h)=L(G)\).
\end{enumerate}
\end{proposition}

\begin{proof}
At a rule application, homomorphic evaluation of the rule template gives
\[
  h^{(\mu(A))}
  \bigl(\rho(\tuple u_1,\ldots,\tuple u_r)\bigr)
  =
  \out_\rho
  \bigl(h^{(\mu(B_1))}(\tuple u_1),\ldots,
        h^{(\mu(B_r))}(\tuple u_r)\bigr).
\]
The rank-zero case is immediate.  Induction on derivation height proves~(i),
and the same identity gives the unique lifting in~(ii).  Erasing types gives
one language inclusion, lifting gives the other, and trimming removes only
objects absent from successful derivations.
\end{proof}

For a surviving typed nonterminal \(X=A_{\mathbf p}\), put
\(L_X:=\{\tuple u:X\Rightarrow_{\widetilde G_0}^*\tuple u\}\).
It is nonempty by trimming.

\begin{lemma}[Concrete exposing contexts]
\label{lem:exposing-contexts}
Let \(X=A_{\mathbf p}\) survive in \(\widetilde G_0\), with arity
\(d=\mu(A)\).  There exists an arity-\(d\) sentence context \(E_X\) such that
\(E_X[\tuple u]\in L(G)\) for every \(\tuple u\in L_X\).
If \(\widetilde S\to X\) is a typed start rule, one may choose
\(E_X=\blank_1\).
\end{lemma}

\begin{proof}
Choose a successful typed derivation containing an occurrence of \(X\).  Delete
the subtree at that occurrence and replace the \(d\) output components by the
named holes \(\blank_1,\ldots,\blank_d\), keeping all other subderivations
fixed.  Nondeletion and linearity propagate each hole exactly once to the
root, so the remaining object is a named sentence context.  Splicing any
\(X\)-derivation back into the deleted position proves the claim.  The start
case has no material outside the selected subtree.
\end{proof}

Fix one such context and write it \(\chi(X)\).%
\footnote{If desired, this choice can be made canonical.  Fix any total order
on \(\Sigma\cup\{\blank_1,\ldots,\blank_d\}\), encode named sentence
contexts by their natural finite strings over this extended alphabet, and order
those encodings by shortlex order: first by length and then lexicographically.
Lemma~\ref{lem:exposing-contexts} guarantees that the set of sentence contexts
exposing \(X\) is nonempty, so it has a unique least element in this order; take
that element as \(\chi(X)\).  This choice is used only on the proof side and is
not supplied to the learner.}
For every surviving typed nonterminal reached directly from the typed start
symbol, we fix the canonical choice \(\chi(X)=\blank_1\).

\subsection{Silent typed children and the characteristic sample}

\begin{definition}[Silent typed nonterminal and positive anchor]
\label{def:silent-typed}
A surviving typed nonterminal \(X\) of arity \(d\) is \emph{silent} if
\(L_X=\{(\eps,\ldots,\eps)\}\).

If \(X\) is not silent, put
\(L_X^+:=\{\tuple u\in L_X:\|\tuple u\|_1>0\}\).
Define \(\omega(X)\) to be an element of \(L_X^+\) of minimum total length
\(\|\tuple u\|_1\), choosing, among all minimum-total-length candidates, the
least one with respect to a fixed order.%
\footnote{The choice can be made uniquely canonical.  Fix a total order on the
finite alphabet \(\Sigma\).  For each \(d\ge1\), use a separator symbol
\(\#\notin\Sigma\) to encode a tuple
\(\tuple u=(u_1,\ldots,u_d)\) as
\(u_1\#u_2\#\cdots\#u_d\).  Compare two tuples first by total length
\(\|\tuple u\|_1\), and, when the total lengths agree, lexicographically by
these encodings under a fixed extension of the alphabet order.  For each fixed
total length there are only finitely many candidates, so whenever
\(L_X^+\neq\emptyset\) this order has a unique least element.  This
canonicalization merely makes an auxiliary proof-side choice unique and does
not provide additional input to the learner.}
If \(X\) is silent, put \(\omega(X):=(\eps,\ldots,\eps)\).
With these conventions, \(\omega(X)\) is well-defined.%
\footnote{If \(X\) is not silent, then \(L_X^+\ne\emptyset\).  Indeed,
\(L_X\ne\emptyset\) for every surviving typed nonterminal, and if every
\(\tuple u\in L_X\) had total length zero, then every component would be empty,
so \(L_X=\{(\eps,\ldots,\eps)\}\), contradicting nonsilence.}
\end{definition}

\begin{lemma}[Distributional equivalence of the empty tuple and the positive anchor]
\label{lem:mixed-empty-positive}
Let \(X\) be a surviving typed nonterminal of arity \(d\).  If \(X\) is not
silent and \((\eps,\ldots,\eps)\in L_X\), then
\(\omega(X)\equiv_L^d(\eps,\ldots,\eps)\).
\end{lemma}

\begin{proof}
Write \(X=A_{\mathbf p}\).  By definition \(\omega(X)\in L_X\), and by
assumption \((\eps,\ldots,\eps)\in L_X\).  Proposition~\ref{prop:output-invariants}
therefore implies that the two tuples have the same componentwise \(h\)-type
\(\mathbf p\).

Moreover, the exposing context \(\chi(X)\) accepts every tuple in \(L_X\), so
\(\chi(X)[\omega(X)]\in L\) and
\(\chi(X)[(\eps,\ldots,\eps)]\in L\).  Thus the two tuples have the same
\(h\)-type and share the same accepting context \(\chi(X)\).
Lemma~\ref{lem:shared-context} yields
\(\omega(X)\equiv_L^d(\eps,\ldots,\eps)\).
\end{proof}

\begin{definition}[Silent-child compression]
\label{def:silent-compression}
Keep every fresh start link \(\widetilde S\to S_{(p)}\) unchanged.  For every
other surviving typed rule
\(R:X\to\rho(Y_1,\ldots,Y_r)\), with \(X=A_{\mathbf p}\), delete from the
body every child occurrence \(Y_j\) that is silent and delete from the
template every variable belonging to that child.

Let \(Y_{j_1},\ldots,Y_{j_s}\) be precisely the remaining non-silent child
occurrences, in their original order, and denote the resulting rule by
\(R^+:X\to\rho^+(Y_{j_1},\ldots,Y_{j_s})\).  If all children are silent,
then \(s=0\), so \(R^+\) has rank zero.  Let \(\widetilde G_0^+\) be the
grammar obtained by applying this operation to every surviving typed rule and
then removing duplicate rules.
\end{definition}

\begin{lemma}[Silent-child compression preserves all typed tuple languages]
\label{lem:silent-compression}
The grammar \(\widetilde G_0^+\) is a finite linear nondeleting MCFG with the
same nonterminals and the same fan-outs as \(\widetilde G_0\).  Moreover, for
every surviving typed nonterminal \(X\),
\(L_X(\widetilde G_0^+)=L_X(\widetilde G_0)\).  In particular,
\(L(\widetilde G_0^+)=L(G)\).

If \(X=A_{\mathbf p}\) is the left-hand side of a compressed rule, every tuple
produced by that rule has componentwise \(h\)-type \(\mathbf p\).
\end{lemma}

\begin{proof}
Fix a deleted silent child \(Y=B_{\mathbf q}\).  By definition,
\(L_Y=\{(\eps,\ldots,\eps)\}\).  Proposition~\ref{prop:output-invariants}
therefore says that the componentwise \(h\)-type of this tuple is
\(\mathbf q\).  Since \(h(\eps)=1_M\), every component of \(\mathbf q\)
must be \(1_M\).

Consequently, in every ground application of the original rule, the variables
belonging to \(Y\) are necessarily instantiated by \(\eps\).  Removing those
variables from the template therefore leaves the ground parent tuple itself
unchanged, and in particular leaves all of its componentwise \(h\)-values
unchanged.

We prove, in both directions by induction on derivation height, that the
original and compressed grammars derive the same tuples at every typed
nonterminal.  From the original grammar to the compressed grammar, discard the
subderivations rooted at deleted silent children and apply the induction
hypothesis to every retained child.  Since every deleted child contributes only
empty strings, the corresponding compressed rule produces exactly the same
parent tuple.

Conversely, for each application of a compressed rule, insert at every deleted
child occurrence a derivation tree of its unique all-empty tuple in the
original grammar.  If the same typed nonterminal occurs more than once in the
body, use an independent copy of such a derivation tree for each occurrence.
Together with the induction hypotheses for the retained children, this
reconstructs an application of the original rule with the same parent tuple.

The original rule is linear and nondeleting.  After deleting the complete
variable sets belonging to silent children, every remaining child variable
still occurs exactly once in the full template, so the compressed rule remains
linear and nondeleting.  The fresh start links are unchanged.  Finiteness and
fan-out preservation follow immediately from the construction, and preservation
of the output type follows from the original output-type invariant together
with the all-identity \(h\)-types of the deleted silent children.
\end{proof}

\begin{definition}[Presentation-relative characteristic sample]
\label{def:cs}
For each surviving typed nonterminal \(X\), include its fixed anchor exposure
\(\chi(X)[\omega(X)]\).

For every positive-rank compressed rule with typed left-hand side
\(X=A_{\mathbf p}\),
\(R^+:X\to\rho^+(Y_1,\ldots,Y_r)\), include the rule exposure
\(\chi(X)[\rho^+(\omega(Y_1),\ldots,\omega(Y_r))]\).

For every compressed rank-zero rule \(R^+:X\to\tuple c\) with typed
left-hand side \(X\), include \(\chi(X)[\tuple c]\).

The finite set of all words selected in this way is denoted
\(\CS(\widetilde G_0)\).  The notation suppresses the dependence on the
auxiliary compression \(\widetilde G_0^+\).
\end{definition}

\begin{lemma}[The characteristic sample is finite and positive]
\label{lem:cs-positive}
If \(L(G)\neq\emptyset\), then \(\CS(\widetilde G_0)\) is a finite nonempty
subset of \(L(G)\).
\end{lemma}

\begin{proof}
There are only finitely many surviving typed nonterminals and compressed rules,
so \(\CS(\widetilde G_0)\) is finite.

If \(L(G)\neq\emptyset\), there is a successful derivation containing at least
one surviving typed start child.  Hence at least one anchor exposure
\(\chi(X)[\omega(X)]\) belongs to the characteristic sample, and therefore
\(\CS(\widetilde G_0)\neq\emptyset\).

By Lemma~\ref{lem:exposing-contexts}, each exposing context \(\chi(X)\) accepts
every tuple in \(L_X\).  By Lemma~\ref{lem:silent-compression}, applying a
compressed rule to the fixed child anchors
\(\omega(Y_1),\ldots,\omega(Y_r)\) produces a parent tuple in \(L_X\); the
same statement applies to the right-hand tuple of a rank-zero rule.  Therefore
every word selected in Definition~\ref{def:cs} belongs to \(L(G)\), and hence
\(\CS(\widetilde G_0)\subseteq L(G)\).
\end{proof}

\subsection{Observed tuples and arbitrary-rank witnesses}

Following Definition~\ref{def:tuple-occurrence}, enumerate all tuple
occurrences of arity at most \(f\) in the sample \(K\).
A tuple \(\tuple x\in(\Sigma^*)^d\), \(1\le d\le f\), is called
\emph{observed in \(K\)} if it is the tuple value of one of these occurrences.
The learner uses one nonterminal \([\tuple x]\) for each observed tuple,
together with a fresh start symbol \(\widehat S\).

\begin{definition}[Positive-support composition witness]
\label{def:composition-witness}
Let \(r\ge1\).
A \emph{positive-support composition witness of rank \(r\)} in \(K\)
consists of a parent occurrence \((E_P,\tuple z)\) of arity \(e\le f\)
and an ordered \(r\)-tuple of child occurrences
\(\tuple x_j=(x_{j,1},\ldots,x_{j,d_j})\)
(\(1\le d_j\le f\)), with repetitions allowed, such that
\(\|\tuple x_j\|_1>0\) for every \(j\).

Each parent component is additionally partitioned, in order, into terminal
gaps and intervals labelled by
\(X_j:=\{x_j^1,\ldots,x_j^{d_j}\}\).
Across all parent components, every variable in
\(X_1\uplus\cdots\uplus X_r\) is used exactly once, and the interval labelled
\(x_j^k\) spells exactly the child component \(x_{j,k}\).
For empty intervals occupying the same position, their local order is also
recorded.

Keeping the terminal gaps and replacing each labelled interval by its
corresponding variable determines a unique linear nondeleting template
\(\rho_{\mathfrak b}\) such that
\(\tuple z=\rho_{\mathfrak b}(\tuple x_1,\ldots,\tuple x_r)\).
\end{definition}

\begin{example}
For instance, let the parent tuple be
\(\tuple z=(ab,cd)\), and let the two child tuples be
\(\tuple x_1=(a,c)\) and \(\tuple x_2=(b,d)\).
Partitioning each parent component according to the child components yields
the template
\(\rho_{\mathfrak b}=(x_1^1x_2^1,x_1^2x_2^2)\), and hence
\(\rho_{\mathfrak b}(\tuple x_1,\tuple x_2)=(ab,cd)=\tuple z\).
This is the simplest rank-two positive-support composition witness.
\end{example}

\begin{lemma}[Sample-bounded witness rank]
\label{lem:sample-bounded-rank}
If the parent tuple \(\tuple z\) of a positive-support rank-\(r\) witness
occurs in a sample word \(w\), then
\(r\le\|\tuple z\|_1\le|w|\le\|K\|_+\).
\end{lemma}

\begin{proof}
The witness template is nondeleting, so every component of every child tuple
occurs exactly once in the parent tuple.  Consequently the total parent length
is the sum of the child total lengths plus the number of explicit terminal
occurrences in the template.  Hence
\(\|\tuple z\|_1\ge\sum_{j=1}^r\|\tuple x_j\|_1\ge r\),
where the last inequality uses the positive total length of every child.  Since the components of the
parent occurrence occupy pairwise disjoint intervals of \(w\) (apart from
zero-length ties), \(\|\tuple z\|_1\le|w|\).
\end{proof}

\begin{lemma}[Enumeration of arbitrary-rank witnesses]
\label{lem:general-witness-enumeration}
Fix \(f\) and a finite sample \(K\), and put \(n:=\|K\|_+\).  All
positive-support composition witnesses in \(K\), together with their induced
canonical rule encodings, can be effectively enumerated.  For every fixed
branching cap \(b\), all witnesses of rank at most \(b\) are enumerable in
\(n^{O(fb)}\) time and space.
\end{lemma}

\begin{proof}
If \(K=\emptyset\), there are no witnesses.  Assume \(n\ge1\).  For a word
of length \(m\), an arity-\(d\) occurrence is specified by a permutation and
\(2d\) cut positions, so the total number of occurrences of arity at most
\(f\) is \(n^{O(f)}\).  For fixed rank \(r\), choose one parent
occurrence and an ordered \(r\)-tuple of child occurrences (with repetition
allowed), then assign to each
of at most \(fr\) normalized variables an output component and two endpoints
inside that parent component.  Enumerate the local order of coincident empty
intervals, reject overlapping or value-inconsistent assignments, and serialize
the uniquely determined terminal gaps and template.  The number of candidates
is \(n^{O(fr)}\).  Summing over \(1\le r\le b\) gives the displayed fixed-\(b\)
bound.  Without a fixed cap, Lemma~\ref{lem:sample-bounded-rank} restricts the
search to \(1\le r\le n\), so the enumeration is still finite and effective.
\end{proof}

\begin{corollary}[Rank-two witness enumeration]
\label{cor:binary-witness-enumeration}
For fixed \(f\), all positive-support witnesses of rank at most two are
enumerable in \(\|K\|_+^{O(f)}\) time and space.
\end{corollary}

\begin{lemma}[Exposed compressed rules occur in the enumeration]
\label{lem:exposed-rule-enumerated}
Let
\(R^+:X\to\rho(Y_1,\ldots,Y_r)\)
be a positive-rank rule of \(\widetilde G_0^+\).  If the sample contains the
anchor exposures of its children and its rule exposure, then the witness
enumeration contains a rank-\(r\) witness whose child tuples are
\(\omega(Y_1),\ldots,\omega(Y_r)\), whose parent tuple is
\(\rho(\omega(Y_1),\ldots,\omega(Y_r))\),
and whose induced template is exactly \(\rho\).
\end{lemma}

\begin{proof}
Put
\(\tuple z_R:=\rho(\omega(Y_1),\ldots,\omega(Y_r))\).
The rule exposure is the sample word \(\chi(X)[\tuple z_R]\), so the pair
\((\chi(X),\tuple z_R)\) supplies the required parent tuple occurrence.  For
each body occurrence \(Y_j\), its anchor exposure
\(\chi(Y_j)[\omega(Y_j)]\) supplies a child occurrence with tuple value
\(\omega(Y_j)\).  If the same typed nonterminal occurs more than once in the
body, the same anchor occurrence may be selected repeatedly, as permitted by
Definition~\ref{def:composition-witness}.

Every child of the compressed rule is non-silent, so each selected anchor has
positive total length.  In the ground application of \(\rho\) to these anchors,
each normalized child variable determines a specific (possibly empty) interval
inside one component of \(\tuple z_R\).  Linearity and nondeletion make these
intervals disjoint except for zero-length coincidences and place every variable
exactly once.  Their order in the template supplies the tie order at coincident
empty cuts; the substrings between consecutive variable intervals are uniquely
the explicit terminal gaps of \(\rho\).  Thus the recorded parent and child
occurrences together with this segmentation form a positive-support witness
whose induced template is exactly \(\rho\).  Completeness of
Lemma~\ref{lem:general-witness-enumeration} emits it.
\end{proof}

\begin{definition}[Capped and general canonical hypotheses]
\label{def:learner}
For \(b\ge1\), the capped hypothesis \(\widehat G_{\le b}(K)\) has start symbol
\(\widehat S\), nonterminals \([\tuple x]\) for all observed tuples of arity
at most \(f\), and the following standard MCFG rules:
\begin{enumerate}[label=(\roman*),leftmargin=*]
\item for every \(w\in K\), the start rule
      \(\widehat S\to[(w)]\);
\item for every observed tuple \(\tuple x\), the rank-zero constant rule
      \([\tuple x]\to\tuple x\);
\item for every positive-support witness of rank \(1\le r\le b\), with parent
      \(\tuple z\), child tuples \(\tuple x_1,\ldots,\tuple x_r\), and induced
      template \(\rho\), the rule
      \[
        [\tuple z]\to
        \rho([\tuple x_1],\ldots,[\tuple x_r]);
      \]
\item for observed tuples \(\tuple x,\tuple y\) of the same arity, the unary
      identity-template rule \([\tuple x]\to[\tuple y]\) whenever
      \(h^{(d)}(\tuple x)=h^{(d)}(\tuple y)\) and some concrete arity-\(d\)
      context \(E\) satisfies
      \(E[\tuple x],E[\tuple y]\in K\).
\end{enumerate}
If \(K\ne\emptyset\), define the \emph{general canonical hypothesis} by
\(\widehat G(K):=\widehat G_{\le\|K\|_+}(K)\).
For \(K=\emptyset\), \(\widehat G(K)\) is the grammar with no productive start
rule.  Fixed canonical encoding orders make both constructions deterministic
and set-driven.  The learner receives only \(K,f,h\).
\end{definition}

\begin{remark}[Compressing semantic unary rules]
The definition may add semantic unary rules for many observed pairs.  An
implementation may first partition observed tuples of each arity and
\(h\)-type by the shared-context relation generated from the sample and choose
one representative per resulting class, connecting other members only to that
representative.  This reduces redundant unary rules without changing the
language generated by the hypothesis or any complexity claim below.
\end{remark}

\begin{lemma}[Finiteness and effectiveness of the canonical hypotheses]
\label{lem:learner-effective}
For every finite \(K\) and every \(b\ge1\), both
\(\widehat G_{\le b}(K)\) and \(\widehat G(K)\) are finite standard MCFGs of
fan-out at most \(f\).  They are effectively constructible from \(K,f,h\);
when \(K\ne\emptyset\), every rule of \(\widehat G(K)\) has rank at most
\(\|K\|_+\).
\end{lemma}

\begin{proof}
A finite sample has only finitely many tuple occurrences of every arity at most
\(f\), hence only finitely many observed tuple values and semantic unary-rule
candidates.  Lemma~\ref{lem:general-witness-enumeration} effectively enumerates
all witness rules below any fixed cap, and
Lemma~\ref{lem:sample-bounded-rank} makes the sample-dependent cap
\(\|K\|_+\) sufficient for the general hypothesis.  The start and constant
rules are finite and effective by inspection.  Every nonterminal
\([\tuple x]\) has fan-out equal to the arity of \(\tuple x\), at most \(f\),
and every emitted composition template is linear and nondeleting, hence a
standard MCFG rule.
\end{proof}

\section{Soundness and Exact Reconstruction}
\label{sec:exact}

\subsection{Soundness}
\label{sec:soundness}

\begin{lemma}[Induced child context]
\label{lem:filling-identity}
Let \(\rho\) be a rank-\(r\) linear nondeleting template, let \(E\) be a
sentence context for its output tuple, and fix compatible tuples for all
children except child \(j\).  Replacing every variable of child \(j\) by its
named hole and every other child variable by the corresponding fixed component
produces a well-formed sentence context \(E_j\) satisfying
\[
  E_j[\tuple u_j]
  =E[\rho(\tuple u_1,\ldots,\tuple u_r)].
\]
The identity remains valid after arbitrary replacements of the fixed sibling
tuples.
\end{lemma}

\begin{proof}
Linearity and nondeletion make every variable of child \(j\) occur exactly
once, so every named hole occurs exactly once.  Filling the holes is simply the
composition of the child substitution with the template substitution and then
with \(E\).
\end{proof}

\begin{lemma}[Witnessed arbitrary-rank composition preserves equivalence]
\label{lem:witnessed-composition}
Let \(L\) be \((f,h)\)-tuple-substitutable and let \(\rho\) be a rank-\(r\)
linear nondeleting template whose output and child arities are at most \(f\).
Suppose
\[
  E[\rho(\tuple x_1,\ldots,\tuple x_r)]\in L
\]
and
\[
  \tuple u_j\equiv_L\tuple x_j
  \qquad(1\le j\le r).
\]
Then
\[
  \rho(\tuple u_1,\ldots,\tuple u_r)
  \equiv_L
  \rho(\tuple x_1,\ldots,\tuple x_r).
\]
\end{lemma}

\begin{proof}
Replace the children one at a time.  At stage \(j\), use
Lemma~\ref{lem:filling-identity} with the already replaced children fixed to
\(\tuple u_i\) and the not-yet-replaced children fixed to \(\tuple x_i\).
The resulting induced context accepts \(\tuple x_j\); equality of tuple
distributions therefore makes it accept \(\tuple u_j\).  After \(r\) stages,
\[
  E[\rho(\tuple u_1,\ldots,\tuple u_r)]\in L.
\]
Homomorphic template evaluation and equality of all child \(h\)-types give
equality of the two parent \(h\)-types.  The original and updated parent tuples
thus share the accepting context \(E\), so Lemma~\ref{lem:shared-context}
applies.
\end{proof}

\begin{proposition}[Soundness of every capped hypothesis]
\label{prop:learner-soundness}
Let \(L\) be \((f,h)\)-tuple-substitutable and \(K\subseteq L\).  For every
\(b\ge1\), if
\([\tuple x]\Rightarrow^*\tuple u\) in \(\widehat G_{\le b}(K)\), then
\(\tuple u\equiv_L\tuple x\).  Consequently,
\[
  L(\widehat G_{\le b}(K))\subseteq L,
  \qquad
  L(\widehat G(K))\subseteq L.
\]
\end{proposition}

\begin{proof}
Induct on the derivation.  A constant rule derives its own observed tuple.  A
semantic unary rule is sound by Lemma~\ref{lem:shared-context}, followed by
symmetry and transitivity of \(\equiv_L\).  For a witness rule, the witness
supplies a parent occurrence \((E,\tuple x)\) in \(K\), hence an accepted
parent application in \(L\); apply the induction hypotheses to all children
and Lemma~\ref{lem:witnessed-composition}.  Finally, every start derivation
begins with \(\widehat S\to[(w_0)]\) for some \(w_0\in K\); equivalence of the
derived unary tuple with \((w_0)\), together with the empty unary context,
places the derived word in \(L\).
\end{proof}

\subsection{Completeness and exact reconstruction}
\label{sec:completeness}

\begin{proposition}[Completeness for a fixed nondeleting presentation]
\label{prop:completeness}
Let \(G\) be a finite nondeleting MCFG of fan-out at most \(f\), let
\(L=L(G)\) be \((f,h)\)-tuple-substitutable, and construct
\(\widetilde G_0^+\) and \(\CS(\widetilde G_0)\) as above.  Let \(b_G\) be the
maximum rank of a compressed rule.  If
\[
  \CS(\widetilde G_0)\subseteq K\subseteq L,
\]
then for every \(b\ge\max\{1,b_G\}\),
\[
  L\subseteq L(\widehat G_{\le b}(K)).
\]
Moreover,
\[
  L\subseteq L(\widehat G(K)).
\]
\end{proposition}

\begin{proof}
Put \(\iota(X):=[\omega(X)]\).  If \(\widetilde S\to X\) survives, then
\(\chi(X)=\blank_1\), so its anchor exposure is itself a sample word and the
start rule \(\widehat S\to\iota(X)\) is present.

For a compressed rank-zero rule \(R:X\to\tuple c\), the anchor exposure and
rule exposure share \(\chi(X)\).  By Lemma~\ref{lem:silent-compression},
\(\tuple c\in L_X\), and Proposition~\ref{prop:output-invariants} therefore
gives the same output type.  Hence the semantic unary rule
\(\iota(X)\to[\tuple c]\) is present, followed by the constant rule
\([\tuple c]\to\tuple c\).  In particular, if \(X\) is non-silent while
\(\tuple c=(\eps,\ldots,\eps)\), this is precisely the mixed case isolated
in Lemma~\ref{lem:mixed-empty-positive}; such a state is not compressed merely
because it also admits an all-empty derivation.

For a positive-rank compressed rule
\[
  R:X\to\rho(Y_1,\ldots,Y_r),
\]
put
\[
  \tuple z_R:=\rho(\omega(Y_1),\ldots,\omega(Y_r)).
\]
By Lemma~\ref{lem:silent-compression}, \(\tuple z_R\in L_X\).  Hence the
anchor exposure and rule exposure share \(\chi(X)\) and have the output type
of \(X\), so \(\iota(X)\to[\tuple z_R]\) is present.  By
Lemma~\ref{lem:exposed-rule-enumerated}, the learner also contains
\[
  [\tuple z_R]\to\rho(\iota(Y_1),\ldots,\iota(Y_r))
\]
whenever its cap is at least \(r\).  An induction on a successful
\(\widetilde G_0^+\)-derivation now simulates the entire derivation.
Lemma~\ref{lem:silent-compression} gives
\(L(\widetilde G_0^+)=L\), proving the capped claim.

For the uncapped general learner, every required witness has a rule exposure
word \(w_R\in K\).  Lemma~\ref{lem:sample-bounded-rank} gives
\(r\le|w_R|\le\|K\|_+\), so every required witness lies below the learner's
sample-dependent cap.
\end{proof}

\begin{theorem}[Exact reconstruction for fixed-observation MCFGs]
\label{thm:exact-reconstruction}
Fix \(f\ge1\) and an explicit finite homomorphism \(h:\Sigma^*\to M\).  Let
\[
  L\in\Cmcf{f}{h}.
\]
Then there exists a finite characteristic sample \(S_L\subseteq L\) such that,
for every finite
\[
  S_L\subseteq K\subseteq L,
\]
the general canonical learner satisfies
\[
  L(\widehat G(K))=L.
\]
If \(L\ne\emptyset\), the sample can be chosen as
\(S_L=\CS(\widetilde G_0)\) for a nondeleting proof-side presentation of
\(L\).  For \(L=\emptyset\), take \(S_L=\emptyset\).
\end{theorem}

\begin{proof}
If \(L=\emptyset\), then \(K=\emptyset\) is the only finite set satisfying
\(K\subseteq L\), and \(\widehat G(\emptyset)\) has empty language by
definition.  Assume \(L\ne\emptyset\).  By Definition~\ref{def:class-cmcf}, choose a finite
standard \(f\)-MCFG \(G\) with \(L(G)=L\).  Since \(G\) is finite it has a
finite branching bound.  By the standard normalization theorem
\cite{Kanazawa2019Ogden}, there is an equivalent MCFG \(G_{\mathrm{nd}}\) with
the same fan-out bound and the same branching bound whose rules are
nondeleting.  Apply the refinement, compression, and characteristic-sample
construction above to \(G_{\mathrm{nd}}\).  Proposition~\ref{prop:completeness}
gives
\[
  L\subseteq L(\widehat G(K)),
\]
while Proposition~\ref{prop:learner-soundness} gives the reverse inclusion.
\end{proof}

\begin{corollary}[Semantic identification in the limit]
\label{cor:semantic-identification}
For every fixed \(f\) and fixed explicit \(h\), the full semantic slice
\(\Cmcf{f}{h}\) is semantically identifiable in the limit from positive data
by the set-driven learner
\[
  K\longmapsto\widehat G(K).
\]
No binary, start-separated, reduced, or working-presentation promise is part of
the target class.
\end{corollary}

\begin{proof}
For every nonempty target, the content of every text eventually contains the
finite characteristic sample from Theorem~\ref{thm:exact-reconstruction},
after which all hypothesis languages equal the target.  For the empty target,
every text consists entirely of pauses, so its content is always empty and the
learner always returns the empty language.
\end{proof}

\begin{corollary}[Explanatory identification by a conservative wrapper]
\label{cor:gold-explanatory}
For every fixed \(f\) and fixed explicit \(h\), the full semantic slice
\(\Cmcf{f}{h}\) is identifiable from positive data in the explanatory sense:
there is a computable learner whose hypothesis grammar itself is eventually
constant and generates the target language.
\end{corollary}

\begin{proof}
Run a conservative wrapper around the canonical learner.  Initialize the
current hypothesis as \(H:=\widehat G(\emptyset)\).  Pauses are ignored.  When a positive example \(w\)
is received, keep \(H\) unchanged if \(w\in L(H)\); otherwise recompute
\(H:=\widehat G(K)\) from the current finite content \(K\).  Membership in a
fixed finite MCFG is decidable, so this wrapper is computable.

For a target \(L\in\Cmcf{f}{h}\), every recomputed canonical hypothesis is
sound by Proposition~\ref{prop:learner-soundness}, hence its language is a
subset of \(L\).  If the current hypothesis language is a proper subset of
\(L\), some missing target word eventually appears on every text and forces a
recomputation.  Once the finite characteristic sample of
Theorem~\ref{thm:exact-reconstruction} is contained in the observed content,
the next recomputation, if any, produces a grammar whose language is exactly
\(L\).  From then on every positive example is already accepted and the
wrapper never changes that grammar again.  If a correct grammar is reached
earlier, stabilization occurs earlier.  The empty target is immediate.
\end{proof}

\begin{corollary}[Bounded-branching specialization]
\label{cor:bounded-rank-reconstruction}
Let \(b\ge1\).  Suppose a target \(L\in\Cmcf{f}{h}\) has a
\(b\)-branching \(f\)-MCFG presentation.  Then it has a characteristic sample such that every
\(K\) containing that sample is reconstructed already by the capped learner
\(\widehat G_{\le b}(K)\).
\end{corollary}

\begin{proof}
Normalize the witnessing presentation without increasing either fan-out or
branching rank \cite{Kanazawa2019Ogden}, then apply
Proposition~\ref{prop:completeness} and
Proposition~\ref{prop:learner-soundness}.  Silent-child compression can only
decrease rule rank.
\end{proof}

\section{Data Complexity under a Supplied Observation}
\label{sec:poly-time}

The qualitative learner allows arbitrary finite rule rank and is therefore not
claimed to run in polynomial time on a finite sample.  This section separates
that issue from identifiability.  If the composition rank is capped by a fixed
\(b\), the same witness construction is polynomial; the later single-spine
analysis uses the special case \(b=2\).  Characteristic-data size is a distinct
question and can remain exponential even when hypothesis construction is
polynomial.

\paragraph{Encoding convention and two input regimes.}
We use the grammar-size convention of Definition~\ref{def:poly-time-data} and
canonical finite encodings for all sample-derived objects.  An occurrence is
recorded by its sample-word index, component permutation, and cut positions;
tuple contexts and templates are recorded by arity together with delimited
terminal factors and ordered hole or variable labels.

For an explicitly supplied observation \(h:\Sigma^*\to M\), let
\(\|h\|_{\mathrm{enc}}\) denote the bit length of its multiplication table and
the letter images.  The first theorem below is slicewise in fixed \(f,b,h\).

\begin{theorem}[Slicewise-polynomial construction under a branching cap]
\label{thm:poly}
Fix a fan-out bound \(f\ge1\), a branching cap \(b\ge1\), and an explicit
finite-monoid homomorphism \(h:\Sigma^*\to M\).  From any finite positive
sample \(K\), the capped canonical hypothesis
\[
  \widehat G_{\le b}(K)
\]
can be constructed, including its complete output encoding, in time
\[
  \|K\|_+^{O(fb)}.
\]
Its fan-out is at most \(f\), and every composition-witness rule has rank at
most \(b\) (semantic unary rules and rank-zero constant rules are also
allowed).  The polynomial degree may depend on \(f\) and \(b\).
\end{theorem}

\begin{proof}
Put \(n:=\max(1,\|K\|_+)\).  Tuple occurrences of arity at most \(f\) are
enumerable in \(n^{O(f)}\) time.%
\footnote{An arity-\(d\) occurrence in a word of length \(m\) is specified by
two cut positions for each component together with the order of the named
components, so the number of candidates is at most
\((m+1)^{2d}d!\).  Moreover \(|K|\le\|K\|_+\le n\) and
\(|w|\le n\) for every \(w\in K\).  Hence all occurrences with
\(d\le f\) can be enumerated in \(n^{O(f)}\) time for fixed \(f\).}  All observed constant rules are read from
these tuples.  Semantic unary rules are obtained by grouping occurrences by
arity and concrete-context encoding and comparing componentwise \(h\)-types;
even exhaustive pairwise comparison remains within \(n^{O(f)}\).

By Lemma~\ref{lem:general-witness-enumeration}, all positive-support composition
witnesses of rank at most \(b\), together with their canonical template
encodings, are enumerable in \(n^{O(fb)}\) time and space.  Canonical sorting
removes duplicate nonterminals and rules within the same bound.  Thus the
number of output symbols, rules, and their total encoding length are all
\(n^{O(fb)}\).
\end{proof}

\begin{corollary}[Polynomial construction with an explicitly supplied observation]
\label{cor:poly-supplied-observation}
Fix \(f\) and \(b\).  If the finite observation
\(h:\Sigma^*\to M\) is supplied as part of the input by its multiplication
table and the images of the alphabet symbols, then
\(\widehat G_{\le b}(K)\) is constructible in time polynomial in
\[
  \|K\|_+
  \quad\text{and}\quad
  \|h\|_{\mathrm{enc}},
\]
including output size.
\end{corollary}

\begin{proof}
The preceding construction uses the observation only to evaluate finitely many
substrings and compare their monoid values.  With the multiplication table and
letter images explicitly encoded, these operations are polynomial in
\(\|h\|_{\mathrm{enc}}\); the combinatorial enumeration remains polynomial for
fixed \(f,b\).
\end{proof}

\begin{remark}[Qualitative versus polynomial construction]
The general learner \(\widehat G(K)\) uses the sample-dependent cap
\(\|K\|_+\), which is sufficient for exact reconstruction by
Lemma~\ref{lem:sample-bounded-rank} but is not covered by the preceding
polynomial-time statement.  If a target has a \(b\)-branching presentation,
Corollary~\ref{cor:bounded-rank-reconstruction} shows that the fixed capped
learner \(\widehat G_{\le b}\) already suffices.  In particular, all working
binary targets considered below use \(b=2\).
\end{remark}

\subsection{Exposure size and the polynomial-data boundary}
\label{subsec:exposure-boundary}

\begin{definition}[Exposure size]
\label{def:exposure-size}
Let \(G\) be a reduced working-form presentation witnessing the target
language, and let \(\widetilde{G}_0\) be its trimmed output-type refinement.
Let \(\CS(\widetilde{G}_0)\) be the characteristic sample of
Definition~\ref{def:cs}, based on the fixed anchor tuples \(\omega(X)\) and
exposing contexts \(\chi(X)\).

The \emph{exposure size} of \(G\) relative to \(h\) is
\[
  B_{\mathrm{exp}}(G,h)
  :=
  \max_{w\in\CS(\widetilde{G}_0)}\max\{1,\lvert w\rvert\}.
\]
Thus \(B_{\mathrm{exp}}(G,h)\) is the maximum positive-example size of a word
in the selected characteristic sample.  Since \(G\) is reduced and an anchor
is selected for every surviving typed nonterminal,
\(\CS(\widetilde{G}_0)\neq\emptyset\), so the maximum is defined.

The notation suppresses the dependence on the fixed choices of \(\omega(X)\)
and \(\chi(X)\).
\end{definition}

\begin{definition}[Polynomially exposable presentation class]
\label{def:polynomially-exposable}
Fix \(f\) and \(h\).  A class \(\mathcal G\) of reduced working binary
presentations is \emph{polynomially exposable} if there is a polynomial \(p\)
such that every \(G\in\mathcal G\) admits anchor and exposing-context choices
with
\[
  B_{\mathrm{exp}}(G,h)\le p(|G|).
\]
\end{definition}

\begin{proposition}[Characteristic-sample size under an exposure bound]
\label{prop:cs-exposure-bound}
Let \(G\) be a reduced working binary linear nondeleting MCFG presentation of
fan-out at most \(f\).  Let \(N_{\mathrm{NT}}\) be the number of surviving typed nonterminals and let
\(N_{\mathrm{rule}}\) be the number of surviving compressed typed rules other
than the fresh rules out of \(\widetilde S\).  Then
\[
  |\CS(\widetilde G_0)|\le N_{\mathrm{NT}}+N_{\mathrm{rule}},
  \qquad
  \|\CS(\widetilde G_0)\|_+
  \le (N_{\mathrm{NT}}+N_{\mathrm{rule}})B_{\mathrm{exp}}(G,h).
\]
For fixed \(f,h\), output typing creates only constantly many copies per
original symbol and rule, so
\(N_{\mathrm{NT}}+N_{\mathrm{rule}}=O(|G|)\).  Hence any presentation family
with \(B_{\mathrm{exp}}(G,h)\) polynomial in \(|G|\) has
presentation-relative polynomial characteristic samples.
\end{proposition}

\begin{proof}
Definition~\ref{def:cs} selects at most one word per surviving typed
nonterminal and per surviving typed rule.  Every selected word has size at
most \(B_{\mathrm{exp}}(G,h)\); multiplying the two bounds proves the claim.
\end{proof}

\begin{theorem}[Polynomial exposure implies polynomial data]
\label{thm:polynomial-exposure-implies-data}
For fixed \(f,h\), every polynomially exposable class of reduced working binary
presentations has presentation-relative polynomial characteristic data for the
canonical learner.
\end{theorem}

\begin{proof}
Combine Definition~\ref{def:polynomially-exposable} with
Proposition~\ref{prop:cs-exposure-bound}.  For fixed \(f,h\), the number of
surviving typed symbols and rules is only a constant-factor expansion of the
working binary presentation size.
\end{proof}

\begin{remark}[Relation to polynomial thickness]
Coste and Nicolas use polynomial thickness to isolate short positive witnesses
in a CFG reduction framework \cite{CosteNicolas2019}.  The present exposure
parameter serves a related purpose but is tailored to typed MCFG
reconstruction: it simultaneously controls anchor exposures for typed
nonterminals and concrete exposures of composition rules.
\end{remark}

The remaining issue is therefore the length of the selected anchor and
exposure words.  It is exponential in general, but polynomial for the
single-spine family below.

\subsection{An exponential obstruction for general binary presentations}
\label{subsec:exponential-exposure}

Polynomial hypothesis construction from a given sample does not imply
polynomial positive data.  The underlying long-singleton phenomenon is
classical in substitutability-based CFG learning: Clark and Eyraud already
pointed out that a compact grammar for the singleton \(\{a^{2^n}\}\) can
force an exponentially long characteristic string, and Coste and Nicolas
revisit essentially this family when motivating polynomial-time-and-thick-data
criteria \cite{ClarkEyraud2007,CosteNicolas2019}.  The formulation below is
used for a different purpose: it lifts the same compact-singleton mechanism to
our working MCFG/fixed-observation setting and states the obstruction
uniformly for arbitrary set-driven learners, rather than only for the
canonical learner.  Thus some structural or quantitative restriction
controlling positive examples is unavoidable for presentation-relative
polynomial data on the full binary family.  
\begin{lemma}[For singleton targets, the positive example itself is characteristic]
\label{lem:singleton-characteristic-samples}
Let \((w_n)_{n\ge0}\) be pairwise distinct nonempty words and put
\(L_n:=\{w_n\}\).  Let \(\mathcal A\) be a set-driven learner.

Each \(L_n\) has only two subsets, \(\emptyset\) and \(L_n\).  Among them,
the empty set can be a characteristic sample for \(L_n\) with respect to
\(\mathcal A\) for at most one index \(n\).

Furthermore, if \(\mathcal A\) identifies every \(L_n\) in the limit from
positive data, then \(L_n\) itself is a characteristic sample for every
\(n\).  Consequently, with at most one exception, the unique characteristic
sample for \(L_n\) is \(L_n=\{w_n\}\).
\end{lemma}

\begin{proof}
Since \(L_n=\{w_n\}\) is a singleton, its only subsets are \(\emptyset\)
and \(L_n\).  Hence these are the only possible characteristic samples.

Suppose that \(\emptyset\) is a characteristic sample for \(L_n\).  By the
definition of characteristic sample, every finite \(K\) satisfying
\(\emptyset\subseteq K\subseteq L_n\) must satisfy
\(L(\mathcal A(K))=L_n\).  Taking \(K=\emptyset\) gives
\(L(\mathcal A(\emptyset))=L_n\).

However, \(\mathcal A(\emptyset)\) is one fixed hypothesis independent of the
target index \(n\).  Its language therefore cannot be equal simultaneously
to two distinct singleton languages \(L_n\) and \(L_m\).  Thus the empty set
is characteristic for at most one \(L_n\).

Now suppose that \(\mathcal A\) identifies \(L_n\) in the limit from positive
data.  Consider the text
\(w_n,w_n,w_n,\ldots\).  After the first occurrence of \(w_n\), the set of
observed positive examples is permanently \(\{w_n\}=L_n\).  Since
\(\mathcal A\) is set-driven, every subsequent hypothesis is
\(\mathcal A(L_n)\).  Identification on this text therefore forces
\(L(\mathcal A(L_n))=L_n\).

If we take \(S=L_n\), the only finite \(K\) satisfying
\(S\subseteq K\subseteq L_n\) is \(K=L_n\).  The preceding equality is
therefore exactly the characteristic-sample condition for \(S=L_n\).
Thus \(L_n\) itself is characteristic.  Combining this with the first part
shows that, except for at most one empty-sample exception, the unique
characteristic sample for \(L_n\) is \(L_n=\{w_n\}\).
\end{proof}

\begin{proposition}[Exponential exposure and learner-uniform data obstruction]
\label{prop:general-exponential-exposure}
Fix any explicit finite monoid homomorphism
\(h\colon\{a\}^*\to M\) and any fan-out bound \(f\ge1\).
Then there exists a family \((G_n)_{n\ge0}\) of reduced working binary linear
nondeleting MCFG presentations of fan-out one such that, writing
\(L_n:=L(G_n)\), we have \(|G_n|=\Theta(n\log n)\),
\(L_n=\{a^{2^n}\}\), and \(L_n\in\Cmcf{f}{h}\) for every \(n\ge0\).

Every nonempty positive sample for \(L_n\) is \(L_n\) itself and therefore
has positive size \(2^n\).  In particular, the presentation-relative exposure
of Definition~\ref{def:exposure-size} satisfies
\(B_{\mathrm{exp}}(G_n,h)=2^n\).

Moreover, for every set-driven learner \(\mathcal A\), there is at most one
index \(n\) for which \(\emptyset\) is a characteristic sample for \(L_n\).
For every other \(n\), either \(L_n\) has no characteristic sample for
\(\mathcal A\), or its unique characteristic sample is
\(L_n=\{a^{2^n}\}\), of positive size \(2^n\).  Consequently, if
\(\mathcal A\) identifies all \(L_n\) in the limit from positive data, then,
with at most one exception, the unique characteristic sample for \(L_n\) has
positive size \(2^n\).

Therefore no set-driven learner identifying \(\Cmcf{f}{h}\) can have
characteristic samples uniformly bounded by a polynomial in presentation size
on the full class of reduced working binary presentations.  This obstruction
already occurs at fan-out one.
\end{proposition}

\begin{proof}
For each \(n\ge0\), define \(G_n\) as follows.  Let
\(A_0,\ldots,A_n\) all have fan-out one, and let \(S\) be the start symbol.
Use the terminal rule \(A_0\to(a)\).  For every \(1\le i\le n\), use the
binary rule
\(A_i\to(x^1y^1)(A_{i-1},A_{i-1})\), and finally add the start rule
\(S\to A_n\).

An induction on \(i\) gives
\(L_{A_i}(G_n)=\{(a^{2^i})\}\).  The base case is immediate from
\(A_0\to(a)\).  If
\(L_{A_{i-1}}(G_n)=\{(a^{2^{i-1}})\}\), then the unique rule for \(A_i\)
concatenates the two child components and therefore has the unique output
\(a^{2^{i-1}}a^{2^{i-1}}=a^{2^i}\).  Hence
\(L(G_n)=\{a^{2^n}\}=L_n\).

Every nonterminal is reachable from the start symbol and productive, so
\(G_n\) is reduced.  All nonterminals have fan-out one, and every binary rule
is linear and nondeleting; hence \(G_n\) is a fan-out-one working binary
linear nondeleting MCFG.

There are \(\Theta(n)\) nonterminals and \(\Theta(n)\) rules.  Under the
encoding convention of Definition~\ref{def:poly-time-data}, a nonterminal
reference requires \(\Theta(\log n)\) bits and every rule template has
constant length.  Thus \(|G_n|=\Theta(n\log n)\).%
\footnote{More explicitly, \(|V|=\Theta(n)\), \(|P|=\Theta(n)\), and every
binary template has length \(O(1)\).  Hence the dominant term in the encoding
formula of Definition~\ref{def:poly-time-data} is
\(|P|\log|V|=\Theta(n\log n)\).}

It remains to show \(L_n\in\Cmcf{f}{h}\).  Put \(N:=2^n\).  The grammar
\(G_n\) has fan-out one, so \(L_n\) is a \(1\)-MCFL and hence an
\(f\)-MCFL for every \(f\ge1\).

We prove \((f,h)\)-tuple substitutability.  Let \(1\le d\le f\), and suppose
\(\tuple x,\tuple y\in(\{a\}^*)^d\) have the same \(h\)-type and share an
accepting named arity-\(d\) context \(E\), so that
\(E[\tuple x],E[\tuple y]\in L_n\).  Let \(c_E\) be the total length of the
fixed terminal material in \(E\).  Over the unary alphabet,
\(|E[\tuple z]|=c_E+\|\tuple z\|_1\) for every tuple \(\tuple z\).
Since both fillings lie in \(L_n=\{a^N\}\),
\(c_E+\|\tuple x\|_1=N\) and
\(c_E+\|\tuple y\|_1=N\).  Hence
\(\|\tuple x\|_1=\|\tuple y\|_1\).

Now let \(F\) be any named arity-\(d\) context and let \(c_F\) be the total
length of its fixed terminal material.  Since the alphabet is unary,
\(F[\tuple z]=a^{c_F+\|\tuple z\|_1}\) for every \(\tuple z\).  Therefore
\(F[\tuple x]\in L_n\) if and only if
\(c_F+\|\tuple x\|_1=N\), which, by equality of total tuple lengths, is
equivalent to \(c_F+\|\tuple y\|_1=N\), and hence to
\(F[\tuple y]\in L_n\).  Since \(F\) was arbitrary,
\(\D_{L_n}^{(d)}(\tuple x)=\D_{L_n}^{(d)}(\tuple y)\).
Thus \(L_n\) is \((f,h)\)-tuple-substitutable.  In fact, this argument does
not use the assumption
\(h^{(d)}(\tuple x)=h^{(d)}(\tuple y)\).  Consequently
\(L_n\in\Cmcf{f}{h}\).

Next consider presentation-relative exposure.  Since \(L_n\) is a singleton,
its only nonempty subset is \(L_n\) itself.  Hence every nonempty positive
sample \(K\subseteq L_n\) satisfies \(K=L_n\) and
\(\|K\|_+=|a^{2^n}|=2^n\).

The sample \(\CS(\widetilde G_0)\) of Definition~\ref{def:cs} is a finite
nonempty subset of \(L_n\) by Lemma~\ref{lem:cs-positive}.  It must therefore
be exactly \(L_n=\{a^{2^n}\}\).  Definition~\ref{def:exposure-size} now
gives \(B_{\mathrm{exp}}(G_n,h)=2^n\).

Finally apply Lemma~\ref{lem:singleton-characteristic-samples} with
\(w_n:=a^{2^n}\).  The words \(w_n\) are pairwise distinct.  Hence, for every
set-driven learner \(\mathcal A\), the empty set is characteristic for at
most one \(L_n\); for every other \(n\), if a characteristic sample exists,
it must be \(L_n=\{a^{2^n}\}\) and has positive size \(2^n\).

In particular, if \(\mathcal A\) identifies \(\Cmcf{f}{h}\) from positive
data, then it identifies every \(L_n\in\Cmcf{f}{h}\).  Thus, with at most one
exception, the unique characteristic sample for \(L_n\) has positive size
\(2^n\).

On the other hand, \(|G_n|=\Theta(n\log n)\).  Hence for every polynomial
\(p\), \(2^n>p(|G_n|)\) for all sufficiently large \(n\).  Therefore the
characteristic-sample sizes on this presentation family cannot be uniformly
bounded by any polynomial in presentation size.  It follows that no
set-driven learner identifying \(\Cmcf{f}{h}\) has presentation-relative
polynomial characteristic samples on the full class of reduced working binary
presentations.
\end{proof}

\subsection{Single-spine presentations}
\label{subsec:single-spine}

The preceding obstruction uses two nonlexical children to unfold a short
grammar into exponentially long witnesses.  Single-spine removes this branching
mechanism and reduces the witness-growth analysis to a path on which repeated
typed states can be removed by subtree splicing.

For this quantitative subsection only, we use the \emph{start-collapsed
working refinement}.  If \(G\) is a start-separated working grammar, begin
with the general output-type refinement of Definition~\ref{def:output-refinement}.
For every working start rule \(S\to A\), replace each successful two-step
typed start fragment
\[
  \widetilde S\to S_{(p)}\to A_{(p)}
\]
by the direct typed start rule
\[
  \widetilde S\to A_{(p)},
\]
and delete the typed copies of \(S\).  Because the working start symbol never
occurs as a child and has no other rules, this preserves the generated
language and all tuple languages of the remaining typed nonterminals.  In this
subsection we again denote the trimmed start-collapsed refinement by
\(\widetilde G_0\).  Its nonstart rules are exactly typed one-letter terminal
rules and typed binary linear nondeleting rules, so it is again a working
presentation in the sense needed by the single-spine arguments.  Moreover,
every surviving typed nonterminal has a positive-total-length tuple: every
productive working derivation contains a one-letter terminal leaf and
nondeletion preserves its contribution.  Thus no surviving typed nonterminal
is silent here, and the positive anchors of Definition~\ref{def:silent-typed}
may be chosen to be minimum-total-length anchors.

\begin{lemma}[Completeness is preserved by start collapse]
\label{lem:start-collapse-completeness}
For a reduced working presentation, the completeness argument of
Proposition~\ref{prop:completeness} remains valid when the ordinary typed
refinement is replaced by the start-collapsed refinement used in this
subsection.
\end{lemma}

\begin{proof}
The collapse changes only successful typed start fragments
\(\widetilde S\to S_{(p)}\to A_{(p)}\) into direct links
\(\widetilde S\to A_{(p)}\).  No nonstart rule or remaining nonterminal tuple
language changes.  For every surviving direct start child \(A_{(p)}\), the
empty unary context \(\blank_1\) is an exposing context, so its anchor itself
supplies the start exposure used in Proposition~\ref{prop:completeness}.  All
nonstart simulation steps are unchanged.
\end{proof}

\begin{lemma}[Output typing preserves the single-spine structure]
\label{lem:single-spine-refinement}
If \(G\) is single-spine, then its start-collapsed typed refinement
\(\widetilde G_0\) is single-spine after a typed nonterminal
\(A_{\mathbf p}\) is declared lexical whenever \(A\) is lexical.
\end{lemma}

\begin{proof}
Output typing changes only the finite type indices on nonterminals and the
compatible copies of rules.  It does not change either child occurrence of an
underlying binary rule.  Hence a typed binary rule has at most one nonlexical
child whenever its untyped rule does.
\end{proof}

For a binary rule
\(\rho:A\to(\alpha_1,\ldots,\alpha_e)(B,C)\), put
\[
  t(\rho):=\sum_{i=1}^e |\alpha_i|_{\Sigma},
\]
the number of explicit terminal occurrences in its template tuple.

\begin{lemma}[Additivity of total tuple length]
\label{lem:tuple-length-additivity}
For every binary linear nondeleting rule \(\rho\) and all compatible child
tuples \(\tuple u,\tuple v\),
\[
  \|\rho(\tuple u,\tuple v)\|_1
  =
  \|\tuple u\|_1+\|\tuple v\|_1+t(\rho).
\]
\end{lemma}

\begin{proof}
Every child variable occurs exactly once in the complete template tuple.
Consequently, the components of \(\tuple u\), the components of \(\tuple v\),
and the explicit terminal occurrences of the rule contribute once each to the
total output length, with neither duplication nor deletion.
\end{proof}

For the remainder of this subsection, let \(G\) be reduced and single-spine.
For every surviving typed nonterminal \(X\), choose an anchor
\(\omega_{\min}(X)\in L_X\) of minimum total tuple length.  Among all
successful derivations containing an occurrence of \(X\), choose one whose
exposing context \(\chi_{\min}(X)\) has the minimum number
\(|\chi_{\min}(X)|_{\Sigma}\) of terminal letters outside its named holes.
Ties are broken by any fixed effective order.  These choices are made only in
the existence proof for the characteristic sample and are not supplied to the
learner.  For the remainder of this subsection, Definition~\ref{def:cs} and
\(B_{\mathrm{exp}}(G,h)\) are instantiated with
\(\omega(X)=\omega_{\min}(X)\) and \(\chi(X)=\chi_{\min}(X)\).

\begin{lemma}[No repetition on a shortest anchor spine]
\label{lem:no-repeat-shortest-anchor}
A minimum-total-length anchor for a surviving typed nonterminal \(X\) has a
derivation whose nonlexical spine contains no typed nonterminal more than once.
\end{lemma}

\begin{proof}
Choose a derivation of a minimum-total-length anchor
\(\omega_{\min}(X)\).  Suppose that a typed nonterminal \(Y\) occurs twice on
its nonlexical spine.  Write \(Y^{\uparrow}\) for the upper occurrence and
\(Y^{\downarrow}\) for the lower occurrence.  Let
\[
  \tuple u^{\uparrow}\in L_Y
  \qquad\text{and}\qquad
  \tuple u^{\downarrow}\in L_Y
\]
be the tuples derived at these two occurrences.

Consider the nonempty spine segment from \(Y^{\uparrow}\) to
\(Y^{\downarrow}\).  At each binary rule on this segment, the child continuing
towards \(Y^{\downarrow}\) is nonlexical.  Since the presentation is
single-spine, the off-spine child is therefore lexical and derives a
one-letter tuple.  If the rules on the segment are
\(\rho_1,\ldots,\rho_k\), listed from bottom to top, repeated application of
Lemma~\ref{lem:tuple-length-additivity} gives
\[
  \|\tuple u^{\uparrow}\|_1
  =
  \|\tuple u^{\downarrow}\|_1
  +
  \sum_{i=1}^{k}\bigl(1+t(\rho_i)\bigr).
\]
Here the term \(1\) is the contribution of the lexical sibling at the
corresponding step.  Since the segment is nonempty,
\[
  \|\tuple u^{\uparrow}\|_1
  >
  \|\tuple u^{\downarrow}\|_1.
\]

Replace the subtree rooted at \(Y^{\uparrow}\) by the subtree rooted at
\(Y^{\downarrow}\).  Since \(Y^{\downarrow}\) is a descendant occurrence
of \(Y^{\uparrow}\) on the same spine and both occurrences have the same
typed label \(Y\), the lower subtree is itself a valid derivation rooted at
\(Y\).  Splicing it into the position of \(Y^{\uparrow}\) therefore
produces a well-typed derivation with the same root \(X\); no rule above
\(Y^{\uparrow}\) is changed.  Put
\[
  \delta
  :=
  \|\tuple u^{\uparrow}\|_1
  -
  \|\tuple u^{\downarrow}\|_1
  >0.
\]
At every ancestor of \(Y^{\uparrow}\), the surrounding rule is linear and
nondeleting.  Hence every component of the replaced child tuple occurs
exactly once in the parent output, and
Lemma~\ref{lem:tuple-length-additivity} shows that replacing the child
decreases the total parent length by exactly \(\delta\).  Inducting upwards
through the ancestors, the total length of the root tuple also decreases by
exactly \(\delta\).  In particular, the resulting root tuple is still an
element of \(L_X\), because the root label and every rule above the splice
point are unchanged.  The resulting root tuple is therefore a strictly shorter
anchor for \(X\), contradicting the minimality of
\(\omega_{\min}(X)\).
\end{proof}

\begin{proposition}[Polynomial anchor length]
\label{prop:single-spine-anchor-bound}
Fix \(f\) and \(h\).  There is a constant \(c_{f,h}\) such that, for every
reduced single-spine working presentation \(G\) of fan-out at most \(f\) and
every surviving typed nonterminal \(X\),
\[
  \|\omega_{\min}(X)\|_1
  \le c_{f,h}|G|^2.
\]
\end{proposition}

\begin{proof}
Let \(N_{\mathrm{NT}}\) be the number of surviving typed nonterminals in
\(\widetilde G_0\).  By
Lemma~\ref{lem:no-repeat-shortest-anchor}, the nonlexical spine of the chosen
anchor derivation contains at most \(N_{\mathrm{NT}}\) typed nonterminal
occurrences.

Every binary step that continues along the nonlexical spine has one lexical
sibling, which contributes one terminal letter.  At the bottom of the spine,
a final binary rule may instead have two lexical children, contributing at
most two letters.  In addition, every binary rule \(\rho\) contributes
\(t(\rho)\le |G|\) explicit terminal occurrences.  Thus, using a harmless
uniform overestimate for all spine nodes,
\[
  \|\omega_{\min}(X)\|_1
  \le
  2+N_{\mathrm{NT}}(2+|G|).
\]
For fixed \(f\) and \(h\), the output-type refinement has
\[
  N_{\mathrm{NT}}=O_{f,h}(|G|)
\]
surviving typed nonterminals.  Consequently,
\[
  \|\omega_{\min}(X)\|_1
  =O_{f,h}(|G|^2).
\]
An anchor obtained directly from a terminal rule has total length one and
satisfies the same bound.
\end{proof}

\begin{lemma}[No repetition on a minimum exposure path]
\label{lem:no-repeat-minimum-exposure}
For every surviving typed nonterminal \(X\), a minimum-terminal-size exposing
derivation can be chosen so that the path from the typed start child to the
selected occurrence of \(X\) contains no repeated typed nonterminal.
\end{lemma}

\begin{proof}
Suppose that a typed nonterminal \(Y\) occurs twice on the selected path,
with an upper occurrence \(Y^{\uparrow}\) and a lower occurrence
\(Y^{\downarrow}\).  The lower occurrence lies on the path from
\(Y^{\uparrow}\) to the selected occurrence of \(X\), so the subtree rooted
at \(Y^{\downarrow}\) still contains that selected occurrence of \(X\).
Since the two roots have the same typed label \(Y\), replacing the subtree
rooted at \(Y^{\uparrow}\) by the one rooted at \(Y^{\downarrow}\) therefore
produces another valid successful typed derivation in which the same selected
occurrence of \(X\) survives.  Hence the resulting derivation still exposes
\(X\).

It remains to compare the terminal size of the exposing contexts.  A repeated
\(Y\) on a proper derivation path cannot be lexical, since a lexical
nonterminal has only one-letter terminal rules and therefore cannot properly
dominate another occurrence.  Thus, at every binary step of the removed path
segment, the child continuing toward \(Y^{\downarrow}\) is nonlexical.  By
the single-spine condition, the off-path child at each such step is lexical
and contributes exactly one terminal letter.  Moreover, every explicit
terminal occurrence in the corresponding rule template lies outside the
selected \(X\)-subtree and therefore belongs to the exposing context.

Consequently, if the removed rules are
\(\rho_1,\ldots,\rho_k\), the splice decreases the number of terminal
letters outside the named holes by
\[
  \sum_{i=1}^{k}\bigl(1+t(\rho_i)\bigr)>0.
\]
This contradicts the minimality of \(\chi_{\min}(X)\).
\end{proof}

\begin{proposition}[Polynomial exposing-context length]
\label{prop:single-spine-context-bound}
Fix \(f\) and \(h\).  There is a constant \(c'_{f,h}\) such that, for every
reduced single-spine working presentation \(G\) of fan-out at most \(f\) and
every surviving typed nonterminal \(X\),
\[
  |\chi_{\min}(X)|_{\Sigma}
  \le c'_{f,h}|G|^2.
\]
\end{proposition}

\begin{proof}
By Lemma~\ref{lem:no-repeat-minimum-exposure}, the selected path has
\(O_{f,h}(|G|)\) typed nonterminal occurrences.  Every ordinary path step is
binary, has a lexical off-path child, and contributes at most \(1+|G|\)
terminal occurrences outside the named holes.  Thus their total contribution
is \(O_{f,h}(|G|^2)\).

There is only one configuration in which an off-path sibling need not be
lexical.  To see this, consider a binary node on the selected path and call
the child continuing along that path the path child.  If the path child is
not the selected occurrence of \(X\) itself, then it properly dominates
\(X\) and hence cannot be lexical.  The single-spine condition therefore
forces the off-path sibling to be lexical.

At the final binary step immediately above \(X\), the path child is \(X\)
itself.  If \(X\) is nonlexical, the off-path sibling is again forced to be
lexical.  Thus a nonlexical off-path sibling can occur only when the selected
occurrence \(X\) is lexical, and only at this final binary step.

Let \(Z\) be this possible nonlexical sibling.  The subtree rooted at \(Z\)
does not contain the selected occurrence of \(X\), so replacing it by any
derivation of the same typed nonterminal preserves the exposing context for
\(X\).  By minimality of the chosen exposing derivation, we may therefore
choose this subtree to derive a minimum anchor of \(Z\).
Proposition~\ref{prop:single-spine-anchor-bound} bounds its terminal
contribution by
\[
  O_{f,h}(|G|^2).
\]
If \(X\) is nonlexical, the exceptional final configuration cannot occur.

Combining this possible contribution with the
\(O_{f,h}(|G|^2)\) contribution of the ordinary path steps gives
\[
  |\chi_{\min}(X)|_{\Sigma}=O_{f,h}(|G|^2).
\]
\end{proof}

\begin{theorem}[Quantitative polynomial time and data for supplied observations]
\label{thm:single-spine-polynomial-exposure}
Fix a fan-out bound \(f\), let \(h:\Sigma^*\to M\) be an explicit finite
monoid homomorphism, and put \(m:=|M|\).  For every reduced single-spine working
binary linear nondeleting MCFG presentation \(G\) of fan-out at most \(f\), the
minimum anchors and exposing contexts selected above satisfy
\[
  B_{\mathrm{exp}}(G,h)=O_f(m^f|G|^2),
  \qquad
  \|\CS(\widetilde G_0)\|_+=O_f(m^{3f}|G|^3).
\]
Consequently, for every fixed fan-out bound \(f\), an explicitly supplied
observation gives two distinct polynomial guarantees.  Hypothesis construction
is polynomial jointly in \(\|K\|_+\) and \(\|h\|_{\mathrm{enc}}\), by
Corollary~\ref{cor:poly-supplied-observation} with the working branching cap
\(b=2\), while the
presentation-relative characteristic-data bound is polynomial in
\(m=|M|\) and \(|G|\) as displayed above.
\end{theorem}

\begin{proof}
The complete output-type refinement has at most
\[
  1+|V|m^f=O_f(m^f|G|)
\]
typed nonterminals.  Each binary rule has at most
\(m^{\mu(B)+\mu(C)}\le m^{2f}\) typed copies.  A working start rule has at most
\(m\) direct typed start copies in the start-collapsed refinement, and a
one-letter terminal rule has one typed copy.  Hence the trimmed refinement has
\[
  N_{\mathrm{NT}}=O_f(m^f|G|),
  \qquad
  N_{\mathrm{rule}}=O_f(m^{2f}|G|).
\]

The proofs of Propositions~\ref{prop:single-spine-anchor-bound}
and~\ref{prop:single-spine-context-bound} use the observation only through
\(N_{\mathrm{NT}}\).  Substituting the displayed bound there gives, for every
surviving typed nonterminal \(X\),
\[
  \|\omega_{\min}(X)\|_1=O_f(m^f|G|^2),
  \qquad
  |\chi_{\min}(X)|_{\Sigma}=O_f(m^f|G|^2).
\]
For a typed binary rule, at most one child is nonlexical; lexical anchors have
length one.  Lemma~\ref{lem:tuple-length-additivity} therefore gives the same
\(O_f(m^f|G|^2)\) bound for every binary exposure, proving the stated bound on
\(B_{\mathrm{exp}}\).

Finally Proposition~\ref{prop:cs-exposure-bound} gives
\[
  \|\CS(\widetilde G_0)\|_+
  \le
  (N_{\mathrm{NT}}+N_{\mathrm{rule}})B_{\mathrm{exp}}(G,h)
  =O_f(m^{3f}|G|^3).
\]
Exact reconstruction is Theorem~\ref{thm:exact-reconstruction}.  The runtime
statement is Corollary~\ref{cor:poly-supplied-observation}; it is measured
against the explicit observation encoding \(\|h\|_{\mathrm{enc}}\).  The data
statement proved here is different: it is presentation-relative and polynomial
in the observation cardinality \(m\) and the witnessing-presentation size
\(|G|\).
\end{proof}

The presentations displayed above for \(L_3\) and \(L_{\times}\) are
single-spine: in each recursive or top rule, all children except the continuing
nonterminal are lexical.  Together with Corollaries~\ref{cor:l3-in-fixed-observation-class}
and~\ref{cor:cross-serial-in-fixed-observation-class},
Theorem~\ref{thm:single-spine-polynomial-exposure} therefore gives polynomial
time and data for both examples under their fixed envelope morphisms.

\section{Observation as Semantic Side Information}
\label{sec:fixed-observation}

The preceding sections supplied the target observation \(h\).  We now vary this
information: first only a bound on observation size is supplied, and then no
bound at all.  We study both identifiability and uniform characteristic-data
complexity.

\begin{definition}[Intrinsic finite-observation size]
\label{def:intrinsic-observation-size}
Fix a fan-out bound \(f\ge1\).  For a language \(L\subseteq\Sigma^*\), define
its \emph{intrinsic finite-observation size} by
\[
  \obs_f(L)
  :=
  \min\left\{
    |\operatorname{im}(h)|
    \;\middle|\;
    \begin{array}{l}
    h:\Sigma^*\to M\text{ is an explicit finite-monoid homomorphism,}\\
    L\text{ is }(f,h)\text{-tuple-substitutable}
    \end{array}
  \right\},
\]
whenever such a finite observation exists, and put
\(\obs_f(L)=\infty\) otherwise.  The image size is used rather than the
ambient codomain size because replacing \(M\) by \(\operatorname{im}(h)\)
does not change equality of observed types.  Thus \(\obs_f(L)\) is a
language-relative semantic parameter and does not depend on a witnessing MCFG
presentation.
\end{definition}

\begin{remark}[Two immediate readings of intrinsic observation size]
\label{rem:intrinsic-observation-readings}
The case \(\obs_f(L)=1\) is exactly unguarded tuple substitutability, since the
one-element monoid makes every pair of tuples observation-equivalent.  More
generally, if a minimum-size witnessing observation of size
\(r=\obs_f(L)<\infty\) is supplied for a reduced single-spine presentation
\(G\), Theorem~\ref{thm:single-spine-polynomial-exposure} gives a
presentation-relative characteristic sample of positive size
\(O_f(r^{3f}|G|^3)\).
\end{remark}

\begin{definition}[Unbounded latent-observation union]
\label{def:no-advice-union}
Fix an alphabet \(\Sigma\) and a fan-out bound \(f\).  We now allow the target
to belong to any fixed-observation slice, but we do not tell the learner which
finite morphism witnesses this membership.  For this reason, define the
\emph{unbounded latent-observation union}
\[
  \Cmcf{f}{\bullet}
  :=
  \bigcup_h \Cmcf{f}{h},
\]
where the union ranges over all explicit finite monoid homomorphisms
\[
  h:\Sigma^*\to M.
\]
A learner for \(\Cmcf{f}{\bullet}\) receives only positive data and the fixed
fan-out bound; in particular, it receives neither a witnessing morphism nor a
bound on the size of its codomain monoid.
\end{definition}

\begin{lemma}[Ascending-chain obstruction]
\label{lem:ascending-chain-obstruction}
Let a language class \(\mathcal C\) contain languages
\[
  L_1\subsetneq L_2\subsetneq\cdots
\]
and their union \(L_\infty=\bigcup_{k\ge1}L_k\).  Then \(\mathcal C\) is not
identifiable in the limit from positive data.
\end{lemma}

\begin{proof}
Assume that a learner \(\mathcal A\) identifies every language in
\(\mathcal C\).  We construct a text for \(L_\infty\) on which
\(\mathcal A\) outputs infinitely many different hypothesis languages.
Fix an enumeration \(v_1,v_2,\ldots\) of \(L_\infty\).  Suppose a finite prefix
\(\sigma_s\) has been constructed and its content is contained in some
\(L_{k_s}\).  Continue \(\sigma_s\) with elements of a text for \(L_{k_s}\).
The resulting infinite continuation is a text for \(L_{k_s}\), so at some
finite extension the learner must output a hypothesis whose language is
\(L_{k_s}\).  Then append an element of
\(L_{k_s+1}\setminus L_{k_s}\), and append \(v_s\) if it has not yet appeared.
The new finite content is contained in some later \(L_{k_{s+1}}\), and the
construction can be repeated.

Every \(v_s\) eventually appears, so the resulting sequence is a text for
\(L_\infty\).  At the selected stages, however, the learner outputs the
strictly increasing approximant languages \(L_{k_s}\), and therefore cannot
converge semantically on this text.  This contradicts identification of
\(L_\infty\).
\end{proof}

\begin{proposition}[Regular languages lie in the latent-observation union]
\label{prop:regular-in-latent-union}
For every \(f\ge1\),
\[
  \mathrm{REG}\subseteq\Cmcf{f}{\bullet}.
\]
\end{proposition}

\begin{proof}
Let \(L\) be regular and let \(\eta_L:\Sigma^*\to M_L\) be its syntactic
morphism.  The monoid \(M_L\) is finite, and \(L\) is an \(f\)-MCFL for every
\(f\ge1\).  If \(\eta_L^{(d)}(\tuple x)=\eta_L^{(d)}(\tuple y)\), then in any
named sentence context one may replace the tuple components one at a time.
Syntactic congruence is two-sided, so each replacement preserves membership in
\(L\).  Hence the complete tuple distributions are equal; the shared-context
premise is not even needed.  Thus \(L\in\Cmcf{f}{\eta_L}\).
\end{proof}

The following is the standard superfinite obstruction underlying Gold-style
and finite-tell-tale impossibility arguments.

\begin{theorem}[Non-identifiability with unbounded latent observation]
\label{thm:no-advice-nonidentifiability}
Fix a nonempty finite alphabet \(\Sigma\).  For every fixed \(f\ge1\),
the latent-observation union \(\Cmcf{f}{\bullet}\) is not semantically
identifiable in the limit from positive data.
\end{theorem}

\begin{proof}
By Proposition~\ref{prop:regular-in-latent-union}, the class contains
\[
  R_k:=\{a,a^2,\ldots,a^k\}\qquad(k\ge1)
\]
and their union \(R_\infty=a^+\).  Since
\(R_1\subsetneq R_2\subsetneq\cdots\),
Lemma~\ref{lem:ascending-chain-obstruction} applies.
\end{proof}

\begin{definition}[Bounded-size observation union]
\label{def:bounded-observation-union}
Fix the finite alphabet \(\Sigma\), a fan-out bound \(f\ge1\), and \(k\ge1\).
Define
\[
  \Cmcf{f}{\le k}:=\bigcup_h \Cmcf{f}{h},
\]
where the union ranges over all finite monoid homomorphisms
\(h:\Sigma^*\to M\) with \(|M|\le k\).
\end{definition}

\begin{proposition}[Intrinsic observation size and bounded latent slices]
\label{prop:intrinsic-size-bounded-slice}
Fix a finite alphabet \(\Sigma\), a fan-out bound \(f\ge1\), and \(k\ge1\).
For every \(f\)-MCFL \(L\subseteq\Sigma^*\),
\[
  \obs_f(L)\le k
  \quad\Longleftrightarrow\quad
  L\in\Cmcf{f}{\le k}.
\]
Equivalently,
\[
  \Cmcf{f}{\le k}
  =
  \{L\in f\text{-}\mathrm{MCFL}:\obs_f(L)\le k\}.
\]
\end{proposition}

\begin{proof}
Suppose first that \(L\in\Cmcf{f}{\le k}\).  Then
\(L\in\Cmcf{f}{h}\) for some explicit finite homomorphism
\[
  h:\Sigma^*\to M
\]
with \(|M|\le k\).  Restricting the codomain to the image of \(h\) does not
change equality of observation values, and therefore does not change
\((f,h)\)-tuple substitutability.  Hence
\[
  \obs_f(L)\le |\operatorname{im}(h)|\le |M|\le k.
\]

Conversely, suppose that \(L\) is an \(f\)-MCFL and
\(\obs_f(L)\le k\).  By Definition~\ref{def:intrinsic-observation-size},
there is an explicit finite homomorphism
\[
  h:\Sigma^*\to M
\]
such that \(L\) is \((f,h)\)-tuple-substitutable and
\(|\operatorname{im}(h)|\le k\).  Replace the codomain of \(h\) by the
submonoid \(\operatorname{im}(h)\).  The resulting explicit homomorphism has
codomain size at most \(k\) and induces exactly the same equality relation on
observed values.  Thus \(L\in\Cmcf{f}{\le k}\).
\end{proof}

\begin{lemma}[Isomorphism invariance of the observation]
\label{lem:iso-invariance}
If \(\varphi:M\to M'\) is an injective monoid homomorphism, then a language is
\((f,h)\)-tuple-substitutable iff it is \((f,\varphi\circ h)\)-tuple-substitutable.
Consequently \(\Cmcf{f}{h}=\Cmcf{f}{\varphi\circ h}\).
\end{lemma}

\begin{proof}
Definition~\ref{def:tuple-substitutable} uses \(h\) only through equalities of
componentwise values.  Since \(\varphi\) is injective,
\(h(u)=h(v)\) iff \(\varphi(h(u))=\varphi(h(v))\).  The grammar-theoretic part of
membership is unchanged.
\end{proof}

\begin{theorem}[Bounded observation size restores identifiability]
\label{thm:bounded-union-identifiable}
For every fixed finite alphabet \(\Sigma\), fan-out bound \(f\ge1\), and
\(k\ge1\), there is an explicit finite monoid homomorphism
\(H=H_{\Sigma,k}:\Sigma^*\to M_H\) such that
\[
  \Cmcf{f}{\le k}\subseteq\Cmcf{f}{H}.
\]
Consequently the canonical learner with parameters \((f,H)\) identifies
\(\Cmcf{f}{\le k}\) in the limit from positive data.
\end{theorem}

\begin{proof}
Up to isomorphism there are only finitely many monoids of cardinality at most
\(k\); fix explicit representatives \(M_1,\ldots,M_t\).  Since \(\Sigma^*\) is
free, a homomorphism \(\Sigma^*\to M_j\) is uniquely determined by an arbitrary
map \(\Sigma\to M_j\), so there are at most \(t k^{|\Sigma|}\) homomorphisms into
the representatives.  Enumerate them as
\(h_i:\Sigma^*\to M_{j_i}\), \(1\le i\le r\), and define
\[
  M_H:=M_{j_1}\times\cdots\times M_{j_r},
  \qquad
  H:=(h_1,\ldots,h_r):\Sigma^*\to M_H,
\]
with componentwise multiplication.  The multiplication table of \(M_H\) and the
letter values \(H(a)\), \(a\in\Sigma\), are computable from the factor tables, so
\(H\) is explicit.

Let \(h:\Sigma^*\to M\) be any explicit homomorphism with \(|M|\le k\).  Choose
an isomorphism \(\varphi:M\to M_j\) onto a representative.  Then
\(\varphi\circ h\) appears in the enumeration, say \(\varphi\circ h=h_i\).  The
homomorphism
\[
  \pi:=\varphi^{-1}\circ\mathrm{pr}_i:M_H\to M
\]
satisfies \(\pi\circ H=h\).  Hence \(h\preceq H\), and
Proposition~\ref{prop:h-refinement-monotonicity} gives
\(\Cmcf{f}{h}\subseteq\Cmcf{f}{H}\).  Taking the union over all \(|M|\le k\)
yields \(\Cmcf{f}{\le k}\subseteq\Cmcf{f}{H}\).  By
Corollary~\ref{cor:semantic-identification}, the canonical learner for the fixed morphism \(H\)
identifies \(\Cmcf{f}{H}\).  Identification in the limit is a universal
statement over targets and texts, so the same learner identifies the subclass
\(\Cmcf{f}{\le k}\).
\end{proof}

\begin{theorem}[Polynomial time and data without a supplied target observation]
\label{thm:bounded-union-single-spine-poly}
Fix a finite alphabet \(\Sigma\), a fan-out bound \(f\ge1\), and an observation
size bound \(k\ge1\).  Let \(\mathcal G^{\mathrm{ss}}_{f,\le k}\) be the class of
all reduced single-spine working binary linear nondeleting MCFG presentations
\(G\) over \(\Sigma\), of fan-out at most \(f\), such that
\[
  L(G)\in\Cmcf{f}{h}
\]
for some explicit finite monoid homomorphism \(h:\Sigma^*\to M\) with
\(|M|\le k\).

There is a single set-driven learner \(\mathcal A_{\Sigma,f,k}\), depending
only on \((\Sigma,f,k)\) and not on the witnessing target morphism \(h\), that
identifies \(\mathcal L(\mathcal G^{\mathrm{ss}}_{f,\le k})\) in polynomial
time and polynomial data.  More precisely, if
\[
  H_{\Sigma,k}:\Sigma^*\to M_H
\]
is the universal product morphism of
Theorem~\ref{thm:bounded-union-identifiable}, and \(\widetilde G_0\) denotes
the trimmed output-type refinement with respect to \(H_{\Sigma,k}\), then the
rank-two capped learner for \((f,H_{\Sigma,k})\) satisfies, for every
\(G\in\mathcal G^{\mathrm{ss}}_{f,\le k}\),
\[
  B_{\mathrm{exp}}(G,H_{\Sigma,k})
  =O_{\Sigma,f,k}(|G|^2),
  \qquad
  \|\CS(\widetilde G_0)\|_+
  =O_{\Sigma,f,k}(|G|^3),
\]
and constructs its rank-two capped hypothesis from every finite sample \(K\) in
\[
  \|K\|_+^{O(f)}
\]
time, including output size.  Thus the target-specific observation morphism is
not required for polynomial time-and-data identification on this bounded
single-spine union.
\end{theorem}

\begin{proof}
Let \(G\in\mathcal G^{\mathrm{ss}}_{f,\le k}\).  By definition there is an
explicit homomorphism \(h:\Sigma^*\to M\), with \(|M|\le k\), such that
\(L(G)\in\Cmcf{f}{h}\).  Theorem~\ref{thm:bounded-union-identifiable} gives
\(h\preceq H_{\Sigma,k}\), and
Proposition~\ref{prop:h-refinement-monotonicity} therefore yields
\[
  L(G)\in\Cmcf{f}{H_{\Sigma,k}}.
\]
The same presentation \(G\) remains reduced, working, and single-spine; only
the semantic observation used by the learner has been refined.

Now \(\Sigma\), \(f\), and \(k\) are fixed, so the finite morphism
\(H_{\Sigma,k}\) is a fixed external parameter.  Applying
Theorem~\ref{thm:single-spine-polynomial-exposure} with
\(h=H_{\Sigma,k}\) gives the stated quadratic exposure bound and cubic
characteristic-sample bound, with constants depending only on
\((\Sigma,f,k)\).  Theorem~\ref{thm:poly}, applied with \(b=2\), gives the
\(\|K\|_+^{O(f)}\) hypothesis-construction bound.  These are precisely the two
requirements of Definition~\ref{def:poly-time-data} for the presentation class
\(\mathcal G^{\mathrm{ss}}_{f,\le k}\).  Since the learner uses only the fixed
universal morphism \(H_{\Sigma,k}\), it is independent of the particular
witnessing target morphism \(h\).
\end{proof}

\begin{remark}[Dependence on the observation-size bound]
\label{rem:bounded-slice-hidden-k}
The constants hidden in \(O_{\Sigma,f,k}(\cdot)\) are not polynomially uniform
in \(k\).  They arise through the universal refinement \(H_{\Sigma,k}\).
Corollary~\ref{cor:universal-observation-size-lower-bound} shows that every
explicitly \(k\)-universal finite observation has codomain size at least
\[
  \Lambda_k=\operatorname{lcm}(1,\ldots,k)=\exp(\Theta(k)).
\]
Thus at least exponential dependence on \(k\) is unavoidable for an explicitly
universal observation; the concrete full product \(H_{\Sigma,k}\) may be
larger still.  Proposition~\ref{prop:universal-learner-exponential-k-data}
shows how this dependence appears in characteristic data for the canonical
universal-refinement learner.
\end{remark}

\subsection{Bounded and unbounded latent observations}
\label{subsec:uniform-observation-bound}

The preceding theorem is polynomial only slicewise in fixed \(k\).  We now
study what happens when the observation-size bound itself is treated as a
parameter, first for the universal observation and then for arbitrary
set-driven learners.

\begin{definition}[Bounded-observation congruence]
\label{def:bounded-observation-congruence}
Fix a nonempty finite alphabet \(\Sigma\) and \(k\ge1\).  For
\(u,v\in\Sigma^*\), write
\[
\begin{aligned}
  u\equiv_k v
  \quad\Longleftrightarrow\quad&
  h(u)=h(v)\\
  &\text{for every finite-monoid homomorphism }h:\Sigma^*\to M
    \text{ with }|M|\le k.
\end{aligned}
\]
Equivalently, \(u\equiv_k v\) says that the monoid identity \(u\approx v\),
with the letters of \(\Sigma\) viewed as variables, holds in every monoid of
cardinality at most \(k\).  The relation \(\equiv_k\) is a finite-index monoid
congruence.
\end{definition}

\begin{proposition}[Universal product realizes the bounded-observation congruence]
\label{prop:bounded-congruence-kernel}
For the product morphism \(H_{\Sigma,k}\) of
Theorem~\ref{thm:bounded-union-identifiable},
\[
  u\equiv_k v
  \quad\Longleftrightarrow\quad
  H_{\Sigma,k}(u)=H_{\Sigma,k}(v).
\]
Thus \(\Sigma^*/{\equiv_k}\) is isomorphic to the image of
\(H_{\Sigma,k}\).
\end{proposition}

\begin{proof}
The coordinates of \(H_{\Sigma,k}\) enumerate, up to isomorphic codomains,
all homomorphisms from \(\Sigma^*\) to monoids of cardinality at most \(k\).
Equality in the product is therefore exactly simultaneous equality under all
such homomorphisms.  The quotient by the kernel of a monoid homomorphism is
isomorphic to its image.
\end{proof}

\begin{proposition}[Exponential index of the bounded-observation congruence]
\label{prop:bounded-congruence-index-lower-bound}
Let \(\Sigma\ne\emptyset\), put
\[
  \Lambda_k:=\operatorname{lcm}(1,2,\ldots,k),
\]
and let \(k\ge1\).  Then
\[
  [\Sigma^*:\!\equiv_k]\ge \Lambda_k.
\]
In fact, for any fixed letter \(a\in\Sigma\), the words
\[
  \eps,a,a^2,\ldots,a^{\Lambda_k-1}
\]
lie in pairwise distinct \(\equiv_k\)-classes.  Consequently
\(\log[\Sigma^*:\!\equiv_k]=\Omega(k)\).
\end{proposition}

\begin{proof}
Take \(0\le r<s<\Lambda_k\).  Since \(0<s-r<\Lambda_k\), the integer
\(s-r\) is not divisible by every \(n\le k\); otherwise it would be divisible
by their least common multiple.  Choose \(n\le k\) with \(n\nmid(s-r)\), let
\(C_n=\langle g\rangle\) be the cyclic group of order \(n\), and map \(a\) to
\(g\) and every other letter to the identity.  Then \(a^r\) and \(a^s\) have
different images, so \(a^r\not\equiv_k a^s\).  The final asymptotic statement
uses \(\log\Lambda_k=\Theta(k)\), equivalently the classical estimate for
Chebyshev's second function; see, e.g., \cite{Nair1982}.
\end{proof}

\begin{definition}[Universal bounded observation]
\label{def:universal-bounded-observation}
Fix a nonempty finite alphabet \(\Sigma\) and \(k\ge1\).  An explicit finite
homomorphism
\[
  U:\Sigma^*\to N
\]
is \emph{\(k\)-universal} if every explicit homomorphism
\(h:\Sigma^*\to M\) with \(|M|\le k\) factors through \(U\); equivalently,
\(h\preceq U\) for every such \(h\).  The morphism \(H_{\Sigma,k}\) constructed
in Theorem~\ref{thm:bounded-union-identifiable} is \(k\)-universal.
\end{definition}

\begin{corollary}[Exponential size is unavoidable for explicit universal observations]
\label{cor:universal-observation-size-lower-bound}
If \(U:\Sigma^*\to N\) is a \(k\)-universal finite observation, then
\[
  |N|\ge [\Sigma^*:\!\equiv_k]\ge \operatorname{lcm}(1,2,\ldots,k).
\]
Hence no explicitly materialized \(k\)-universal observation can have codomain
size polynomial in \(k\).
\end{corollary}

\begin{proof}
If \(U(u)=U(v)\), then every \(h\preceq U\) also satisfies \(h(u)=h(v)\).
Thus \(\ker U\subseteq\equiv_k\).  Hence
\(|\operatorname{im}U|=[\Sigma^*:\ker U]\ge[\Sigma^*:\equiv_k]\), and
\(|N|\ge|\operatorname{im}U|\).  Apply
Proposition~\ref{prop:bounded-congruence-index-lower-bound}.
\end{proof}

\begin{lemma}[The unary positive language has intrinsic observation size two]
\label{lem:a-plus-observation-size}
Let \(\Sigma=\{a\}\).  For every fan-out bound \(f\ge1\),
\[
  \obs_f(a^+)=2.
\]
More explicitly, let \(M_+=\{e,z\}\) be the two-element monoid with identity
\(e\) and absorbing idempotent \(z\), and let
\[
  h_+:\{a\}^*\to M_+,
  \qquad
  h_+(a)=z.
\]
Then \(a^+\) is \((f,h_+)\)-tuple-substitutable for every \(f\ge1\), whereas
it is not tuple-substitutable for the one-element observation already at
arity one.
\end{lemma}

\begin{proof}
The morphism \(h_+\) records exactly whether a word is empty.  Suppose
\(\tuple x,\tuple y\) have equal componentwise \(h_+\)-type.  Then, for every
component index \(i\),
\[
  x_i=\eps \quad\Longleftrightarrow\quad y_i=\eps.
\]
For any named sentence context \(F\), the filled unary word \(F[\tuple x]\)
is empty exactly when the fixed terminal part of \(F\) is empty and every
inserted component \(x_i\) is empty.  By the displayed equivalence, this holds
exactly when \(F[\tuple y]=\eps\).  Hence
\[
  F[\tuple x]\in a^+
  \quad\Longleftrightarrow\quad
  F[\tuple y]\in a^+,
\]
so the tuple distributions agree; the shared-context premise is not even
needed.  Thus \(a^+\) is \((f,h_+)\)-tuple-substitutable for every \(f\).

For the lower bound, under the one-element observation take the unary tuples
\(x=\eps\) and \(y=a\).  They share the accepting context \(a\blank_1\), since
both \(a\blank_1[\eps]=a\) and \(a\blank_1[a]=aa\) belong to \(a^+\).  Their
distributions differ, however, because the context \(\blank_1\) rejects
\(\eps\) and accepts \(a\).  Therefore the one-element observation does not
suffice, and \(\obs_f(a^+)=2\).
\end{proof}

\begin{proposition}[The universal-refinement learner has exponential data dependence on \(k\)]
\label{prop:universal-learner-exponential-k-data}
Let \(\Sigma=\{a\}\), \(f=1\), and let \(G_+\) be the fixed reduced
single-spine working grammar
\[
  S\to A,\qquad
  A\to(a),\qquad
  B\to(a),\qquad
  A\to(x^1y^1)(A,B),
\]
which generates \(a^+\).  For each \(k\ge2\), run the canonical learner with
the universal product observation \(H_{\{a\},k}\).  Every characteristic
sample for \(a^+\) for this learner contains a word of length at least
\[
  \Lambda_k=\operatorname{lcm}(1,2,\ldots,k).
\]
In particular its positive size is at least \(\Lambda_k=\exp(\Theta(k))\),
even though the witnessing grammar \(G_+\) has constant size and the target
has the constant intrinsic observation size \(\obs_1(a^+)=2\).
\end{proposition}

\begin{proof}
By Lemma~\ref{lem:a-plus-observation-size}, \(a^+\) belongs to a fixed
two-element observation slice.  Hence, for every \(k\ge2\), the universal
morphism \(H_{\{a\},k}\) refines a witnessing observation for \(a^+\), so
Proposition~\ref{prop:h-refinement-monotonicity} puts the target in the
corresponding universal slice.  The exact-reconstruction theorem therefore
supplies a finite characteristic sample for the canonical learner on every
such slice.

By Proposition~\ref{prop:bounded-congruence-kernel} and
Proposition~\ref{prop:bounded-congruence-index-lower-bound}, the words
\(\eps,a,\ldots,a^{\Lambda_k-1}\) have pairwise different
\(H_{\{a\},k}\)-values.  Let \(K\subseteq a^+\) be finite and suppose every
word of \(K\) has length smaller than \(\Lambda_k\).  Every unary tuple
observed in \(K\) is \(a^i\) with \(0\le i<\Lambda_k\).  Therefore two
distinct observed unary tuple values never have the same universal type, and
the canonical construction introduces no nontrivial unit rule between
distinct tuple values.

The new constant rules, including \([\eps]\to\eps\), do not invalidate the
length argument because positive-support witness rules never use a zero-length
child.  Consider a witness rule with parent label \(a^q\) and child labels
\(a^{q_1},\ldots,a^{q_r}\).  Then every \(q_j\ge1\), and the witness equation
gives
\[
  q=t(\rho)+\sum_{j=1}^r q_j.
\]
If \(r\ge2\), or if \(t(\rho)>0\), every child label is strictly shorter than
the parent.  The only remaining case is \(r=1\), \(t(\rho)=0\), and
\(q_1=q\).  At fan-out one over the unary alphabet, the induced nondeleting
template is then the identity template \(x^1\), so the rule is a vacuous
self-loop and may be deleted.  Together with the absence of nontrivial semantic
unary rules proved above, every remaining productive dependency therefore
strictly decreases the positive integer labeling its nonterminal.  Rank-zero
constant rules create no dependency edges.  Hence the productive dependency
graph is acyclic, and the finite hypothesis grammar generates a finite
language.  It therefore cannot generate the infinite target \(a^+\).

If \(S_+\) were a characteristic sample for \(a^+\) with all words shorter
than \(\Lambda_k\), taking \(K=S_+\) would contradict the defining property of
a characteristic sample.  Thus some word in every characteristic sample has
length at least \(\Lambda_k\), and the positive-size lower bound follows.
\end{proof}

\subsection{A learner-uniform ambient-slice obstruction}
\label{subsec:uniform-k-data-lower-bound}

The preceding lower bounds are observation-specific: they arise from the size
and kernel structure of a universal refinement.  We now give a different kind
of result.  It uses only the fact that a sufficiently large bounded latent
slice contains a rich deletion family, and therefore applies uniformly to
every set-driven identifier of that ambient class.  It should not be read as a
claim that the deletion obstruction is specific to finite-monoid observations.

\begin{definition}[Uniform polynomial data in the observation bound]
\label{def:uniform-polynomial-k-data}
Fix a finite alphabet \(\Sigma\) and a fan-out bound \(f\).  A family of
set-driven learners
\[
  (\mathcal A_k)_{k\ge1}
\]
for the bounded-observation single-spine classes
\(\mathcal L(\mathcal G^{\mathrm{ss}}_{f,\le k})\) has
\emph{uniform polynomial data in \(k\)} if there is one polynomial \(p\),
independent of \(k\), such that for every \(k\) and every
\(G\in\mathcal G^{\mathrm{ss}}_{f,\le k}\), the target \(L(G)\) has a
characteristic sample \(S_G\) for \(\mathcal A_k\) satisfying
\[
  \|S_G\|_+\le p(k,|G|).
\]
This strengthens the slicewise notion of
Definition~\ref{def:poly-time-data}, where \(k\) is fixed externally and may be
absorbed into the polynomial constants.
\end{definition}

\begin{lemma}[Fixed-length deletion family]
\label{lem:fixed-length-deletion-family}
Let \(\Gamma=\{0,1\}\), let \(n\ge1\), and put
\[
  F_n:=\Gamma^n,
  \qquad
  F_{n,w}:=\Gamma^n\setminus\{w\}\quad(w\in\Gamma^n),
\]
\[
  k_n:=\frac{n(n+1)}2+2.
\]
For every fan-out bound \(f\ge1\):
\begin{enumerate}[label=(\roman*),leftmargin=*]
\item \(\obs_f(F_n)=1\); in particular \(F_n\) belongs to the
one-element observation slice and hence to \(\Cmcf{f}{\le k_n}\).  Moreover,
\(F_n\) has a reduced single-spine working binary linear nondeleting
presentation \(G_n^{F}\) with \(|G_n^{F}|=O(n\log n)\);
\item for every \(w\in\Gamma^n\),
\(F_{n,w}\in\Cmcf{f}{\le k_n}\), and \(F_{n,w}\) has a reduced
single-spine working binary linear nondeleting presentation of size
\(O(n\log n)\).
\end{enumerate}
\end{lemma}

\begin{proof}
For \(F_n\), let
\[
  \mathbf 1:\Gamma^*\to\{e\}
\]
be the unique homomorphism to the one-element monoid.  Suppose two arity-\(d\)
tuples \(\tuple x,\tuple y\) share an accepting sentence context \(E\).  Since
both fillings lie in \(\Gamma^n\), the outside contribution of \(E\) is the
same and therefore
\[
  \sum_{i=1}^d |x_i|=\sum_{i=1}^d |y_i|.
\]
For every other arity-\(d\) sentence context \(F\), the two fillings
\(F[\tuple x]\) and \(F[\tuple y]\) consequently have the same total length.
Membership in \(\Gamma^n\) depends only on that total length, so
\[
  F[\tuple x]\in F_n
  \quad\Longleftrightarrow\quad
  F[\tuple y]\in F_n.
\]
Thus \(F_n\) is \((f,\mathbf 1)\)-tuple-substitutable for every \(f\ge1\).
Since no finite monoid image can have size below one,
\[
  \obs_f(F_n)=1.
\]
In particular \(F_n\in\Cmcf{f}{\le k_n}\).

A single-spine presentation for \(F_n\) uses nonterminals
\(X_1,\ldots,X_n\), with
\[
  X_1\to(0),\qquad X_1\to(1),
\]
\[
  X_i\to(x^1y^1)(X_{i-1},X_1)
  \qquad(2\le i\le n),
\]
and the start rule \(S\to X_n\).  The nonterminal \(X_1\) is lexical, so each
binary rule has at most one nonlexical child.  The presentation is reduced and
generates exactly all binary words of length \(n\).  It has \(O(n)\)
nonterminals and rules, constant-length templates, and therefore encoding size
\(O(n\log n)\) under Definition~\ref{def:poly-time-data}.

Now fix \(w=c_1\cdots c_n\in\Gamma^n\), and let
\[
  \eta_w:\Gamma^*\to M_w
\]
be the syntactic morphism of the singleton language \(\{w\}\).  The syntactic monoid of the singleton language \(\{w\}\) has at most
\[
  \frac{n(n+1)}2+2=k_n
\]
classes: every word that is not a factor of \(w\) has empty two-sided context
distribution with respect to \(\{w\}\) and hence belongs to one common class,
while every remaining class contains either \(\eps\) or a nonempty factor of
\(w\).  Thus \(|M_w|\le k_n\).
We claim that \(F_{n,w}\) is \((f,\eta_w)\)-tuple-substitutable for every
\(f\).  Let \(\tuple x,\tuple y\) have equal componentwise \(\eta_w\)-type and
share an accepting arity-\(d\) context \(E\).  Since both
\(E[\tuple x]\) and \(E[\tuple y]\) lie in \(F_{n,w}\), both have length
\(n\), and therefore
\[
  \sum_{i=1}^d |x_i|=\sum_{i=1}^d |y_i|.
\]
For any other arity-\(d\) sentence context \(F\), the two fillings
\(F[\tuple x]\) and \(F[\tuple y]\) consequently have the same total length.
Moreover, replacing the named components one at a time preserves syntactic
congruence with respect to \(\{w\}\); hence
\[
  F[\tuple x]=w
  \quad\Longleftrightarrow\quad
  F[\tuple y]=w.
\]
Thus
\[
  F[\tuple x]\in\Gamma^n\setminus\{w\}
  \quad\Longleftrightarrow\quad
  F[\tuple y]\in\Gamma^n\setminus\{w\},
\]
so the tuple distributions are equal.

For completeness, \(F_{n,w}\) also has a linear-size single-spine working
presentation.  Write \(\overline c\) for the other binary letter.  Use lexical
nonterminals \(A_0\to(0)\) and \(A_1\to(1)\).  Let \(E_i\) derive exactly the
prefix \(c_1\cdots c_i\), and let \(D_i\) derive exactly the length-\(i\) words
different from that prefix.  Start with
\[
  E_1\to(c_1),\qquad D_1\to(\overline{c_1}),
\]
and, for the relevant indices, use
\[
  E_i\to(x^1y^1)(E_{i-1},A_{c_i}),
\]
\[
  D_i\to(x^1y^1)(D_{i-1},A_0),\qquad
  D_i\to(x^1y^1)(D_{i-1},A_1),
\]
\[
  D_i\to(x^1y^1)(E_{i-1},A_{\overline{c_i}}).
\]
Take \(S\to D_n\), omit unused \(E_n\), and trim any redundant symbols in the
case \(n=1\).  An induction on \(i\) gives the stated tuple languages.  Every
binary rule again has at most one nonlexical child, and the trimmed grammar is
reduced.  It has \(O(n)\) symbols and rules and hence encoding size
\(O(n\log n)\).
\end{proof}

\begin{theorem}[No uniform polynomial-data learner without target observation]
\label{thm:no-uniform-polynomial-k-data}
Fix the binary alphabet \(\Gamma=\{0,1\}\) and any fan-out bound \(f\ge1\).
Let
\[
  k_n=\frac{n(n+1)}2+2.
\]
For every \(n\) and every set-driven learner \(\mathcal A_{k_n}\) that
identifies
\(\mathcal L(\mathcal G^{\mathrm{ss}}_{f,\le k_n})\) from positive data,
every characteristic sample for the target \(F_n=\Gamma^n\) is exactly
\(F_n\).  Consequently every such characteristic sample has positive size
\[
  \|F_n\|_+=n2^n.
\]
In particular, no family of set-driven learners for the bounded-observation
single-spine slices has uniform polynomial data in \(k\) in the sense of
Definition~\ref{def:uniform-polynomial-k-data}.  This failure already holds for
fan-out-one finite regular languages with witnessing presentations of size
\(O(n\log n)\).
\end{theorem}

\begin{proof}
By Lemma~\ref{lem:fixed-length-deletion-family}, the target \(F_n\) and every
proper sublanguage
\[
  F_{n,w}=F_n\setminus\{w\}
  \qquad(w\in F_n)
\]
belong to the same bounded-observation single-spine slice with parameter
\(k_n\).

Because \(F_n\) is finite and \(\mathcal A_{k_n}\) is set-driven and identifies
it, a text that has displayed every member of \(F_n\) and then repeats them
forces
\[
  L(\mathcal A_{k_n}(F_n))=F_n.
\]
Hence \(F_n\) itself is a characteristic sample, so characteristic samples for
the target do exist.

Let \(S\subseteq F_n\) be any characteristic sample for \(F_n\) for
\(\mathcal A_{k_n}\), and fix \(w\in F_n\).  Suppose \(w\notin S\).  Then
\[
  S\subseteq F_{n,w}\subsetneq F_n.
\]
Because \(F_{n,w}\) is finite and \(\mathcal A_{k_n}\) is set-driven and
identifies it, once a text for \(F_{n,w}\) has displayed every member, its
content remains exactly \(F_{n,w}\); hence necessarily
\[
  L(\mathcal A_{k_n}(F_{n,w}))=F_{n,w}.
\]
On the other hand, the characteristic property of \(S\), applied to the finite
sample \(K=F_{n,w}\), gives
\[
  L(\mathcal A_{k_n}(F_{n,w}))=F_n,
\]
a contradiction.  Thus every \(w\in F_n\) belongs to \(S\), so \(S=F_n\).
Its positive size is \(n2^n\).

Finally,
\[
  k_n=\Theta(n^2),
  \qquad
  |G_n^{F}|=O(n\log n),
\]
where \(G_n^{F}\) is the presentation from
Lemma~\ref{lem:fixed-length-deletion-family}.  Hence a polynomial in
\((k_n,|G_n^{F}|)\) is only polynomial in \(n\), whereas \(n2^n\) is exponential.
Equivalently, the lower bound is
\[
  2^{\Omega(\sqrt{k_n})}
\]
up to a polynomial factor.  Therefore no uniform polynomial
characteristic-data bound can exist.
\end{proof}

Two immediate consequences are worth recording.  Theorem~\ref{thm:no-uniform-polynomial-k-data}
places no restriction on how a set-driven learner represents or chooses an
observation, so adaptive selection cannot remove the barrier.  Moreover,
Lemma~\ref{lem:fixed-length-deletion-family} has \(\obs_f(F_n)=1\) and
\(|G_n^F|=O(n\log n)\) while every characteristic sample in the ambient slice
has size \(n2^n\); hence intrinsic observation size and presentation size alone
cannot bound latent-slice data polynomially, independently of the ambient
slice parameter.

\begin{corollary}[Exponential sample value of supplied slice information]
\label{cor:constant-fiber-advice-exponential-value}
Let \(\mathbf 1:\Gamma^*\to\{e\}\) be the one-element observation and let
\(G_n^{F}\) be the single-spine presentation of \(F_n=\Gamma^n\) from
Lemma~\ref{lem:fixed-length-deletion-family}.  For fixed \(f\), the canonical
learner supplied with \(\mathbf 1\) has a presentation-relative
characteristic sample for \(F_n\) of positive size
\[
  O_f(|G_n^{F}|^3)=O_f(n^3\log^3 n).
\]
By contrast, every set-driven learner that identifies the whole ambient slice
\(\mathcal L(\mathcal G^{\mathrm{ss}}_{f,\le k_n})\), without being supplied a
target slice, requires characteristic-sample size exactly \(n2^n\) for the
same target \(F_n\).  The two bounds concern two different semantic promises: the fixed
one-element slice and the full ambient bounded-observation slice.  Thus the
corollary quantifies the value of the stronger supplied-slice promise; it is
not a within-class lower bound for one fixed learner.
\end{corollary}

\begin{proof}
The first bound is Theorem~\ref{thm:single-spine-polynomial-exposure}
with \(m=1\), together with \(|G_n^{F}|=O(n\log n)\).  The second statement is
Theorem~\ref{thm:no-uniform-polynomial-k-data}.
\end{proof}

\begin{remark}[Scope of the lower bound]
Corollary~\ref{cor:constant-fiber-advice-exponential-value} concerns supplied
semantic slice information, not unrestricted nonuniform advice.  The uniform
lower bound is stronger than the size lower bound for an explicit universal
product: it applies to every set-driven strategy on the full bounded slice,
including strategies that try to infer a coarser target observation.  Removing
the size bound altogether gives the separate non-identifiability result of
Theorem~\ref{thm:no-advice-nonidentifiability}.
\end{remark}

\section{Comparison with Distributional Learning}
\label{sec:comparison}

The fixed-observation condition belongs to the distributional tradition of
grammatical inference, but the organizing question here is the value of a
specified semantic information interface: which merges become safe when a
finite observation is supplied and what fails when the relevant slice is
latent.  This differs both in semantic condition and in information model from
earlier MCFG constructions.

\subsection{Relation to Yoshinaka's multidimensional substitutability}
\label{subsec:comparison-multidimensional}

Yoshinaka developed multidimensional substitutability for positive-data
learning of mildly context-sensitive languages in an earlier conference paper
and in the subsequent journal article
\cite{Yoshinaka2009,Yoshinaka2011}; his preceding work treated substitutable
context-free languages \cite{Yoshinaka2008}.  His learner is the closest
predecessor of the tuple-indexed part of the present construction.  The main
interfaces are summarized below.

\begin{center}
\small
\begin{tabularx}{0.96\textwidth}{>{\raggedright\arraybackslash}p{0.25\textwidth}>{\raggedright\arraybackslash}p{0.33\textwidth}>{\raggedright\arraybackslash}X}
\hline
Axis & Yoshinaka's multidimensional framework & Present framework \\
\hline
Semantic merge condition & multidimensional substitutability on restricted multicontexts & named shared context plus componentwise equality under a supplied finite observation \\
Dimension parameter & fixed nonterminal dimension \(p\) & fixed fan-out \(f\) \\
Rule-rank parameter & fixed \(r\) in the positive-data learning slice \(SL(p,r)\) & no common rank cap in the qualitative theorem; a cap is used only for polynomial construction \\
Information interface & positive data, with no supplied finite observation & positive data plus a finite observation \(h\), or a bound on its size in the latent-slice results \\
Quantitative role & fixed \((p,r)\)-slice & branching cap for construction; polynomial exposure for data \\
Known separation & --- & \(L_3\) lies in a fixed-observation slice but fails \(S(p)\) for every \(p\ge2\) \\
\hline
\end{tabularx}
\end{center}

The two semantic conditions are nevertheless genuinely different.
Proposition~\ref{prop:l3-not-yoshinaka-2d} exhibits a language in the present
fixed-observation slice that violates Yoshinaka's substitution condition
already at dimension two.  We use the journal definition as the formal
comparison point and compare with Section~3.1, p.~1824.  Let
\(
  \mathsf{Ctx}_{\mathrm Y}^{(m)}
  :=\Sigma^*(\blank\Sigma^+)^{m-1}\blank\Sigma^*
\)
and, for \(x=x_0\blank x_1\blank\cdots\blank x_m\), write
\(x\odot(y_1,\ldots,y_m)=x_0y_1x_1\cdots y_mx_m\).

\begin{definition}[Yoshinaka's \(p\)D-substitutability
{\cite[Section~3.1, p.~1824]{Yoshinaka2011}}]
\label{def:yoshinaka-pd-substitutability}
A language \(L\) is \emph{\(p\)D-substitutable} if, for every
\(1\le m\le p\), all \(x_1,x_2\in\mathsf{Ctx}_{\mathrm Y}^{(m)}\), and all
\(\tuple y_1,\tuple y_2\in(\Sigma^+)^m\),
\[
 x_1\odot\tuple y_1,\quad x_1\odot\tuple y_2,\quad
 x_2\odot\tuple y_1\in L
 \quad\Longrightarrow\quad
 x_2\odot\tuple y_2\in L.
\]
Write \(S(p)\) for this semantic class.  Yoshinaka's learned class is
\(SL(p,r)=S(p)\cap L(p,r)\), where \(p\) bounds nonterminal dimension and
\(r\) bounds rule-function rank.
\end{definition}

Thus Yoshinaka's positive-data identification theorem works on each fixed
\((p,r)\)-slice.  By contrast, Corollary~\ref{cor:semantic-identification}
fixes \(f\) and the finite observation \(h\) but takes the union over all
finite target rule ranks; only Theorem~\ref{thm:poly} reintroduces a fixed
branching cap for efficiency.  This difference is orthogonal to the semantic
separation below.

The good-grammar convention in Yoshinaka's paper
(\(\lambda\)-free, non-erasing, non-permuting, and non-merging) is a grammar
normal form, not part of the definition of \(S(p)\).  Yoshinaka's multicontexts
embed into the present named-context domain as the identity-order subdomain
with nonempty tuple components and nonempty factors between consecutive holes.
On this common domain, the present rule is the guarded version: a shared
context licenses a unit rule only when the componentwise \(h\)-types also
agree.  No global inclusion between the full semantic classes is claimed.

\begin{proposition}[Separation from Yoshinaka's multidimensional substitutability]
\label{prop:l3-not-yoshinaka-2d}
Let \(h_{P_3}=h_3^+\) be the fixed transition morphism defined in
Section~\ref{sec:running-examples} for \(P_3=a^+b^+c^+\).
Then \(L_3\in\Cmcf{2}{h_{P_3}}\), whereas
\(L_3\notin S(p)\) for every \(p\ge2\).
In particular, \(L_3\notin SL(2,2)\).
\end{proposition}

\begin{proof}
Corollary~\ref{cor:l3-in-fixed-observation-class} gives
\(L_3\in\Cmcf{2}{h_{P_3}}\).

We show that Yoshinaka's condition already fails at dimension \(2\).
Let \(\tuple x=(a,c)\), \(\tuple y=(a^2b,c^2)\),
\(E=\blank b\blank\), and \(F=\blank abb\blank c\).
These are contexts allowed by
Definition~\ref{def:yoshinaka-pd-substitutability}, and all tuple components
and factors occurring between consecutive holes are nonempty.

Indeed, \(E\odot\tuple x=abc\in L_3\), and
\(E\odot\tuple y=F\odot\tuple x=a^2b^2c^2\in L_3\), whereas
\(F\odot\tuple y=aababbccc\notin L_3\).
Thus \(\tuple x\) and \(\tuple y\) share one successful context but are
separated by another context.  Yoshinaka's multidimensional substitutability
condition therefore fails already for \(m=2\).

Since \(p\)D-substitutability requires the condition for every
\(1\le m\le p\), it follows that \(L_3\notin S(p)\) for every \(p\ge2\).
In particular, \(L_3\notin SL(2,2)\).

\end{proof}

This proposition separates the two conditions on the specific language
\(L_3\).  We make no further general inclusion claim between Yoshinaka's
classes and the full fixed-observation slices considered here.

\paragraph{Comparison with untyped named contexts.}
One may also remove the \(h\)-type guard from the full named-context domain of
the present paper.  Call a language \emph{untyped tuple-substitutable} up to
fan-out \(f\) if, whenever two tuples \(\tuple x,\tuple y\) of the same arity
\(d\le f\) share at least one accepting named context, their full context
distributions are equal.

Untyped tuple substitutability implies \((f,h)\)-tuple substitutability for
every finite observation \(h\).  Indeed, the untyped condition requires
equality of distributions from the existence of a shared accepting context
alone, whereas the \((f,h)\)-condition imposes that conclusion only when one
additionally has \(h^{(d)}(\tuple x)=h^{(d)}(\tuple y)\).  Thus introducing
the \(h\)-type guard weakens the semantic substitutability requirement.

For the zone morphism \(h_{P_3}\), this inclusion is strict.  The tuples
\(\tuple x=(a,c)\) and \(\tuple y=(a^2b,c^2)\) used above share the accepting
context \(E\) but are separated by \(F\), so \(L_3\) is not untyped
tuple-substitutable.  On the other hand, \(h_{P_3}\) assigns different types
to these two tuples, and hence \((f,h_{P_3})\)-tuple substitutability does not
require their distributions to agree.  By
Corollary~\ref{cor:l3-in-fixed-observation-class},
\(L_3\in\Cmcf{f}{h_{P_3}}\) for every \(f\ge2\).
Consequently, for this fixed observation and every \(f\ge2\), the class of
\(f\)-MCFLs satisfying the untyped condition is a proper subclass of
\(\Cmcf{f}{h_{P_3}}\).

\subsection{Relation to the author's context-free fixed-typing theorem}
\label{subsec:comparison-author-cfg}

The author's context-free fixed-typing theorem, proved in the author's 2014
master's thesis at the Graduate University for Advanced Studies
(SOKENDAI), Japan,\footnote{The Graduate University for Advanced Studies,
SOKENDAI, Japan.}
gives typed positive-data reconstruction for CFGs \cite{kuriyamaCFG}.
The present paper has no logical dependence on that result: the MCFG proof
given here is self-contained.  Some high-level ideas carry over---finite
observation guards, finite positive characteristic samples, and distributional
soundness---but the tuple setting introduces additional proof obligations
that have no literal one-dimensional counterpart.

\begin{center}
\small
\begin{tabularx}{0.96\textwidth}{>{\raggedright\arraybackslash}p{0.24\textwidth}>{\raggedright\arraybackslash}p{0.31\textwidth}>{\raggedright\arraybackslash}X}
\hline
Issue & CFG reconstruction & Present MCFG reconstruction \\
\hline
Observed object & string occurrence & named tuple occurrence, including empty components and component order \\
Composition witness & constituent boundary & occurrence-sensitive placement of every child component in a linear template \\
Soundness step & one-dimensional substitution & induced child contexts and child-by-child replacement (Lemma~\ref{lem:filling-identity}) \\
Empty behavior & epsilon derivations & silent typed states and silent-child compression (Lemma~\ref{lem:silent-compression}) \\
Rule rank & native CFG production structure & arbitrary finite MCFG rank; required rank is bounded by the positive sample (Lemma~\ref{lem:sample-bounded-rank}) \\
Fixed fan-out & not applicable & binarization may increase fan-out, so direct arbitrary-rank witnesses are needed \\
Quantitative data & linear/thickness regimes & polynomial exposure as the abstract criterion; single-spine as one sufficient regime \\
\hline
\end{tabularx}
\end{center}

In particular, the MCFG theorem cannot be obtained by simply replacing strings
with tuples and then appealing to a binary normal form.  The occurrence model,
silent-child compression, induced child contexts, and sample-dependent rank
cap are used directly in the exact-reconstruction proof, while the
single-spine anchor/context analysis is a separate quantitative argument.

The arbitrary-rank construction is essential at fixed fan-out: nondeleting
normalization preserves fan-out and branching rank \cite{Kanazawa2019Ogden},
whereas LCFRS binarization can increase fan-out and need not preserve the
original fixed-\(f\) class \cite{GomezRodriguezEtAl2009}.  Thus a common
branching cap is imposed only for the polynomial-construction result, not for
qualitative identification.

\subsection{Finite-kernel, finite-context, and query-based distributional learning}
\label{subsec:comparison-fkp}

In distributional learning, finite-kernel and finite-context methods avoid
treating infinitely many substrings or tuples separately by using finitely many
representative objects or contexts to capture their distributional behavior.
There are corresponding approaches for MCFGs
\cite{Yoshinaka2010,ClarkYoshinaka2014,ClarkYoshinaka2016}.

The output-type refinement used here has a related singleton-kernel flavor.%
\footnote{That is, within each typed component, one representative tuple can
stand for the distributional behavior of the tuples belonging to that
component.}
Within each surviving typed component, all tuples have the same \(h\)-type and
occur in a common exposing context with the representative tuple.
Tuple substitutability therefore makes them have the same sentence-context
distribution.
In this sense, one representative tuple captures the distributional role of
the whole typed component.

We do not, however, claim the stronger finite-context characterization found
in the finite-context literature \cite{KanazawaYoshinaka2017}: namely, that
the set of all tuples passing a fixed finite collection of context tests is
exactly the tuple language generated by a given nonterminal.
Nor is the problem of deciding such a congruential property from a supplied
grammar presentation part of the present learning theorem.
Here \((f,h)\)-tuple substitutability is assumed as a semantic promise on the
target language; the learner is not required to decide it from a presentation
\cite{Clark2017,KanazawaKappe2019}.

\paragraph{Membership queries versus fixed finite observation.}
Earlier MCFG finite-kernel and finite-context learning uses membership
information
\cite{Yoshinaka2010,ClarkYoshinaka2014,ClarkYoshinaka2016}.
Related work also studies learners equipped with membership and equivalence
queries
\cite{YoshinakaClark2012,KanazawaYoshinaka2021,KanazawaYoshinaka2023}.

The finite observation \(h\) supplied in the present paper is not a membership
oracle.
It classifies strings into finitely many types and tells the learner whether
two strings or tuples carry the same finite-state information.
Knowing \(h(w)\) therefore does not, in general, tell us whether \(w\) belongs
to the target language.
Instead, the learner uses this finite-state information to control whether two
tuples observed in positive data are even candidates for identification with
the same grammatical state.

Conversely, access to a membership oracle does not automatically provide a
finite congruence \(h\) specifying which tuples may safely be identified.
Fixed finite observation and membership queries are therefore different forms
of side information, and we claim no direct conversion between one interface
and a number of queries to the other.

\section{Conclusion}
\label{sec:conclusion}

Finite observation provides a precise semantic interface for positive-data
learning of bounded-fan-out MCFGs.  For every fixed fan-out \(f\) and supplied
finite observation \(h\), occurrence-sensitive witnesses, output-type
refinement, silent-child compression, and sample-bounded arbitrary rule rank
give exact reconstruction on the full semantic slice \(\Cmcf{f}{h}\), with no
binary or working-presentation restriction on the qualitative target class.
The canonical learner is set-driven and semantically convergent, and a
conservative wrapper additionally gives explanatory stabilization of the
hypothesis grammar.

The quantitative boundary is most naturally expressed through exposure.
Polynomially exposable presentation classes have polynomial
presentation-relative characteristic data.  Single-spine working
presentations form one explicit sufficient regime: lexical children can be
inlined to obtain an equivalent rank-at-most-one MCFG, and the resulting
nonlexical derivation structure supports polynomial anchor and context bounds.
At fan-out one this specializes to linear context-free languages, while at
higher fan-out it still covers genuinely non-context-free examples.

The observation interface itself admits three regimes.  A supplied observation
supports fixed-slice learning; a fixed bound on an unknown observation can be
compiled into a universal finite refinement slicewise; and without any bound
the latent-observation union is not identifiable from positive data.  The
universal refinement has an observation-specific quantitative cost: every
explicitly \(k\)-universal observation has at least exponential size in \(k\),
and the canonical universal-refinement learner exhibits corresponding long
exposures even for \(a^+\).  Separately, sufficiently large bounded latent
slices contain deletion families that rule out any characteristic-data bound
polynomial jointly in the observation bound and presentation size (with the
constructed family giving \(2^{\Omega(\sqrt{k})}\), up to a polynomial
factor).  The latter is an ambient-class obstruction, not a claim that the
deletion phenomenon is specific to finite monoids.

Natural examples delimit the semantic class.  MIX lies in the trivial
one-element slice, whereas the copy language is a fan-out-two MCFL with no
finite witnessing observation.  Thus finite observation is neither a disguised
Parikh-only condition nor a characterization of all bounded-fan-out MCFLs.
Finally, intrinsic observation size by itself is not a sufficient statistic for
latent-slice data complexity: the learning promise also matters, specifically
whether a witnessing slice is supplied or must be resolved inside a larger
ambient class.

\section*{Acknowledgements}

The author is sincerely grateful to Makoto Kanazawa for valuable discussions
and advice, and to Ryo Yoshinaka for encouraging the completion of research
that had remained unfinished since the author's doctoral studies.

\end{document}